\documentclass[a4paper,USenglish]{lipics}
\usepackage{amssymb,amsmath,amsthm,enumerate,graphicx,float,lmodern,xparse, mathtools}
\usepackage[noend]{algpseudocode}
\usepackage{algorithm}
\usepackage{hyperref}
\usepackage{fixltx2e}
\usepackage{microtype}
\usepackage{longtable}

\title{Practical Algorithms for Linear Boolean-width\footnote{The research of the third author was partially funded by the
Networks programme, funded by the Dutch Ministry of Education, Culture and
Science through the Netherlands Organisation for Scientific Research.}}

\author[1]{Chiel B. ten Brinke}
\author[1]{Frank J. P. van Houten}
\author[1]{Hans L. Bodlaender}
\affil[1]{Department of Computer Science, Utrecht University\\
    PO Box 80.089, 3508 TB Utrecht, The Netherlands\\
\texttt{CtenBrinke@gmail.com, Frankv@nhouten.com, H.L.Bodlaender@uu.nl}}
\authorrunning{Ch.\,B. ten Brinke, F.\,J.\,P. van Houten and H.\,L. Bodlaender} 

\subjclass{G.2.2 [Discrete Mathematics]: Graph Theory --- Graph algorithms; F.2.2 [Analysis of Algorithms and Problem Complexity]: Nonnumerical Algorithms and Problems --- Computations on discrete structures}
\keywords{graph decomposition, boolean-width, heuristics, exact algorithms, vertex subset problems}

\Copyright{Chiel B. ten Brinke, Frank J. P. van Houten and Hans L. Bodlaender}

\serieslogo{}
\volumeinfo
  {Billy Editor and Bill Editors}
  {2}
  {Conference title on which this volume is based on}
  {1}
  {1}
  {1}
\EventShortName{}
\DOI{10.4230/LIPIcs.xxx.yyy.p}

\DeclareMathOperator*{\argmin}{argmin}

\NewDocumentCommand\set{mg}{%
    \ensuremath{\left\lbrace #1 \IfNoValueTF{#2}{}{\, \middle|\, #2} \right\rbrace}%
}
\NewDocumentCommand\powerset{mg}{%
    \ensuremath{\mathcal{P}(#1)}
}

\newtheorem*{claim}{\textbf{Claim}}

\newcommand{\twopartdef}[4]{%
    \left\{
        \begin{array}{ll}
            #1 & \mbox{if } #2 \\
            #3 & \mbox{if } #4
        \end{array}
    \right.
}

\newcommand{\un}{\mathcal{UN}}
\DeclareMathOperator{\booldim}{bool-dim}
\DeclareMathOperator{\bw}{boolw}

\DeclareMathOperator{\pw}{pw}
\DeclareMathOperator{\lbw}{lboolw}

\newcommand{\cut}[1]{(#1, \overline{#1})}
\newcommand{\nec}[1]{nec(\equiv_{#1}^d)}

\newif\ifarchive{}
\archivetrue{}

\begin{document}
\maketitle

\begin{abstract}
In this paper, we give a number of new exact algorithms and heuristics to compute linear boolean decompositions, and experimentally evaluate these algorithms.
The experimental evaluation shows that significant improvements can be made with respect to running time
without increasing the width of the generated decompositions.
We also evaluated dynamic programming algorithms on linear boolean decompositions for several vertex subset problems.
This evaluation shows that such algorithms are often much faster (up to several orders of magnitude) compared to theoretical worst case bounds.
\end{abstract}

\section{Introduction}
\label{section:introduction}
Boolean-width is a recently introduced graph parameter~\cite{boolean-width}.
Similarly to treewidth and other parameters, it measures some structural complexity of a graph.
Many NP-hard problems on graphs become easy if some graph parameter is small.
We need a derived structure which captures the necessary information of a graph in order to exploit such a small parameter.
In the case of boolean-width, this is a binary partition tree, referred to as the decomposition tree.
However, computing an optimal decomposition tree is usually a hard problem in itself.
A common approach to bypass this problem is to use heuristics to compute decompositions with a low boolean-width.

Algorithms for computing boolean decompositions have been studied before in
~\cite{vatshelle, UpperBoundsBoolw, Practicalaspects, Hvidevold},
but in this paper we study the specific case of linear boolean decompositions, which are considered in~\cite{Belmonte201354, UpperBoundsBoolw, Practicalaspects}.
Linear decompositions are easier to compute and the theoretical running time of algorithms for solving practical problems is lower on linear decompositions than on tree shaped ones.
For instance, vertex subset problems can be solved in $O^*( {nec}^3)$ due to a dynamic programming algorithm by Bui-Xuan et al.~\cite{fastdynamicprogramming},
but this can be improved to $O^*({nec}^2)$ for linear decompositions.
Here, $nec$ is the number of d-neighborhood equivalence classes, i.e., the maximum size of the dynamic programming table.

We first give an exact algorithm for computing optimal linear boolean decompositions, improving upon existing algorithms,
and subsequently investigate several new heuristics through experiments, improving upon the work by Sharmin~\cite[Chapter 8]{Practicalaspects}.
We then study the practical relevance of these algorithms in a set of experiments by solving an instance of a vertex subset problem,
investigating the number of equivalence classes compared to the theoretical worst case bounds.

\section{Preliminaries}
\label{section:preliminaries}

A \emph{graph} $G = (V, E)$ of size $n$ is a pair consisting of a set of $n$ \emph{vertices} $V$ and a set of \emph{edges} $E$. 
The \emph{neighborhood} of a vertex $v \in V$ is denoted by $N(v)$.
For a subset $A \subseteq V$ we denote the neighborhood by $N(A) = \bigcup_{v \in A} N(v)$.
In this paper we only consider simple, undirected graphs and assume we are given a total ordering on the vertices of a graph $G$. 
For a subset $A \subseteq V$ we denote the \emph{complement} by $\overline{A} = V \setminus A$.
A partition $\cut{A}$ of $V$ is called a \emph{cut} of the graph.
Each cut $\cut{A}$ of $G$ induces a bipartite subgraph $G[A, \overline{A}]$.

The \emph{neighborhood across a cut} $\cut{A}$ for a subset $X \subseteq A$ is defined as $N(X) \cap \overline{A}$.

\begin{definition}[Unions of neighborhoods]
Let $G = (V, E)$ be a graph and $A \subseteq V$. We define the set of \emph{unions of neighborhoods across a cut $(A, \overline{A})$} as
\[
    \un(A) = \set{ N(X) \cap \overline{A} }{X \subseteq A}.
\]
\end{definition} 

The number of unions of neighborhoods is symmetric for a cut $\cut{A}$, i.e., $|\un(A)| = |\un(\overline{A})|$~\cite[Theorem 1.2.3]{booleanMatrixTheory}.
Furthermore, for any cut $\cut{A}$ of a graph $G$ it holds that $|\un(A)| = \#\mathcal{MIS}(G[A, \overline{A}])$, where $\#\mathcal{MIS}(G)$ is the number of maximal independent sets in $G$~\cite[Theorem 3.5.5]{vatshelle}.

\begin{definition}[Decomposition tree] 
A \emph{decomposition tree} of a graph $G = (V, E)$ is a pair $(T, \delta)$, where $T$ is a full binary tree and $\delta$ is a bijection between the leaves of $T$ and vertices of $V$.
If $a$ is a node and $L$ are its leaves, we write $\delta(a) = \bigcup_{l \in L} \delta(l)$.
So, for the root node $r$ of $T$ it holds that $\delta(r) = V$.
Furthermore, if nodes $a$ and $b$ are children of a node $w$, then $(\delta(a), \delta(b))$ is a partition of $\delta(w)$.
\end{definition} 

In this paper we consider a special type of decompositions, namely \emph{linear decompositions}.

\begin{definition}[Linear decomposition] 
A \emph{linear decomposition}, or \emph{caterpillar decomposition}, is a decomposition tree $(T, \delta)$
where $T$ is a full binary tree and for which each internal node of $T$ has at least one leaf as a child.
We can define such a linear decomposition through a linear ordering $\pi = \pi_1,\dots,\pi_n$ of the vertices of $G$ by letting $\delta$ map the $i$-th leaf of $T$ to $\pi_i$.
\end{definition}

\begin{definition}[Boolean-width]
Let $G = (V, E)$ be a graph and $A \subseteq V$. 
The \emph{boolean dimension of $A$} is a function $\booldim : 2^{V} \to \mathbb{R}$.
    \[
        \booldim(A) =\log_2 |\un(A)|.
    \]

Let $(T, \delta)$ be a decomposition of a graph $G$. We define the \emph{boolean-width} of $(T, \delta)$ as the maximum boolean dimension over all cuts induced by nodes of $(T, \delta)$.
    \[
        \bw(T, \delta) = \max\limits_{w \in T} \booldim(\delta(w))
    \]

The boolean-width of $G$ is defined as the minimum boolean-width over all possible full decompositions of $G$,
while the \emph{linear boolean-width} of a graph $G = (V(G), E(G))$ of size $n$ is defined as the the minimum boolean-width over all linear decompositions of $G$.
    \[
        \bw(G) = \min\limits_{\ (T, \delta)\ of\ G} \bw(T, \delta) 
    \]
    \[
        \lbw(G) = \min\limits_{linear\ (T, \delta)\ of\ G} \bw(T, \delta) 
    \]
\end{definition}

It is known that for any graph $G$ it holds that $\bw(G) \leq treewidth(G) + 1$~\cite[Theorem 4.2.8]{vatshelle}.
The linear variant of treewidth is called  \emph{pathwidth}~\cite{Pathwidth}, or $\pw$ for short.

\begin{theorem}
[Appendix~\ref{prf:pathwidth}]
\label{th:pathwidth}
For any graph $G$ it holds that $\lbw(G) \leq \pw(G) + 1$.
\end{theorem}

The algorithms in this paper make extensive use of sets and set operations, 
which can be implemented efficiently by using bitsets.
By using a mapping from vertices to bitsets that represent the neighborhood of a vertex we can store the adjacency matrix of a graph efficiently.
We assume that bitset operations take $O(n)$ time and need $O(n)$ space, even though in practice this may come closer to $O(1)$.
If one assumes that these requirements are constant, several time and space bounds in this paper improve by a factor $n$.

In this paper we assume that the graph $G$ is connected, since if the graph consists of multiple connected components
we can simply compute a linear decomposition for each connected component, after which we glue them together, in any arbitrary order.

\section{Exact Algorithms}
\label{section:exactalgorithms}
We can characterize the problem of finding an optimal linear decomposition by the following recurrence relation, in which $P$ is a function mapping a subset of vertices $A$ to the linear boolean-width of the induced subgraph $G[A, \overline{A}]$.

\begin{equation}
\label{eq:lboolw}
    \begin{aligned}
        P(\{v\}) &= |\un(\{v\})| = \twopartdef{1}{N(v) = \emptyset}{2}{N(v) \neq \emptyset}\\
        P(A)     &= \min_{v \in A} \{\max\{|\un(A)|, P(A \setminus \{v\})\}\}
    \end{aligned}
\end{equation}
The boolean-width of the graph $G$ is now given by $\log_2(P(V))$.
Adaptation of existing techniques lead to the following algorithms for linear boolean-width, upon we hereafter improve:
\begin{itemize}
    \item With dynamic programming a running time of $O(2.7284^n)$ is achieved. (See Theorem~\ref{thm:lbw_dp}, Appendix~\ref{subs:existing_exact_algorithms})
    \item With adaptation of the exact algorithm for boolean-width by Vatshelle~\cite{vatshelle}, a running time of $O(n^3 \cdot 2^{n + \lbw(G)})$ is achieved.
    (See Theorem~\ref{thm:lbw_vatshelle}, Appendix~\ref{subs:existing_exact_algorithms})
\end{itemize}

\subsection{Improving the running time}
We present a faster and easier way to precompute for all cuts $A \subseteq V$ the value $|\un(A)|$,
which results in a new algorithm displayed in Algorithm~\ref{alg:incremental-un-exact}.
In the following it is important that the $\un$ sets are implemented as hashmaps, which will only save distinct neighborhoods.

\begin{algorithm}[H]
\caption{Compute $\un(X \cup \set{v})$ given $\un(X)$.}
\label{increment-un}
\begin{algorithmic}[1]
\Procedure{Increment-UN}{$G,X,\un_X,v$}
    \State{$\mathcal{U} \leftarrow \emptyset$}
    \For{$S \in \un_X$}
        \State{$\mathcal{U} \leftarrow \mathcal{U} \cup \set{S \setminus \set{v}}$}
        \State{$\mathcal{U} \leftarrow \mathcal{U} \cup \set{(S \setminus \set{v}) \cup (N(v) \cap (\overline{X} \setminus
        \set{v}))}$}
    \EndFor{}
    \State{\textbf{return} $\mathcal{U}$}
\EndProcedure{}
\end{algorithmic}
\end{algorithm}

\begin{lemma}
[Appendix~\ref{prf:increment-un-correctness}]
\label{increment-un-correctness}
    The procedure Increment-UN is correct and runs in $O(n \cdot |\un_X|)$ time using $O(n \cdot |\un_X|)$ space.
\end{lemma}

\begin{algorithm}[H]
\caption{Return $\lbw(G)$, if it is smaller than $\log K$, otherwise return $\infty$.}
\label{alg:incremental-un-exact}
\begin{algorithmic}[1]
\Procedure{Incremental-UN-exact}{$G,K$}
    \State{$T_{\un}(\emptyset) \leftarrow 0$}
    \State{\Call{Compute-count-UN}{$G,K,T_{\un},\emptyset,\set{\emptyset}$}}
    \State{}
    \State{$P(X) \leftarrow \infty$, for all $X \subseteq V$}
    \State{$P(\emptyset) \leftarrow 0$}
    \State{}
    \For{$i \leftarrow 0, \dots, |V|-1$}
        \For{$X \subseteq V$ of size $i$}
            \For{$v \in V \setminus X$}
                \State{$Y \leftarrow X \cup \set{v}$}
                \If{$P(X) \leq K$}
                    \State{$P(Y) \leftarrow \min(P(Y), \max(T_{\un}(Y), P(X)))$}
                \EndIf{}
            \EndFor{}
        \EndFor{}
    \EndFor{}
    \State{}
    \State{\textbf{return} $\log_2(P(V))$}
\EndProcedure{}

\State{}

\Procedure{Compute-count-UN}{$G,K,T_{\un},X,\un_X$}
    \For{$v \in V \setminus X$}
        \State{$Y \leftarrow X \cup \set{v}$}
        \If{$T_{\un}(Y)$ is not defined}
\label{booldim-pruning1}
            \State{$\un_{Y} \leftarrow$ \Call{Increment-UN}{$G,X,\un_X,v$}}
            \State{$T_{\un}(Y) \leftarrow |\un_{Y}|$}
            \If{$T_{\un}(Y) \leq K$}
\label{booldim-pruning2}
                \State{\Call{Compute-count-UN}{$G,K,T_{\un},Y,\un_{Y}$}}
            \EndIf{}
        \EndIf{}
    \EndFor{}
\EndProcedure{}
\end{algorithmic}
\end{algorithm}

\begin{theorem}
[Appendix~\ref{prf:incremental-un-exact-correctness}]
\label{incremental-un-exact-correctness}
    Given a graph $G$, Algorithm~\ref{alg:incremental-un-exact} can be used to compute $\lbw(G)$
    in $O(n \cdot 2^{n+\lbw(G)})$ time using $O(n \cdot 2^n)$ space.
\end{theorem}

This new algorithm improves upon the time in Theorem~\ref{thm:lbw_vatshelle} by a factor $n^2$, while the space requirements stay the same.
Since the tightest known upperbound for linear boolean-width is $\frac{n}{2} - \frac{n}{143} + O(1)$~\cite{UpperBoundsBoolw}, this
algorithm can be slower than dynamic programming, since $O(2^{n + \frac{n}{2} - \frac{n}{143} + O(1)}) = O(2.8148^{n + O(1)}) \supsetneq O(2.7284^n)$,
but this is very unlikely to happen in practice.

\section{Heuristics}
\label{section:heuristics}
\subsection{Generic form of the heuristics}
The goal when using a heuristic is to find a linear ordering of the vertices in a graph in such a way that the 
decomposition that corresponds to this ordering will be of low boolean-width.
A basic strategy to accomplish this is to start the ordering with some vertex and then by some 
selection criteria append a new vertex to the ordering that has not been appended yet.
\ifarchive{}
This strategy is used in heuristics introduced by Sharmin~\cite[Chapter 8]{Practicalaspects}, 
and a similar approach is shown in Algorithm~\ref{alg:generatevertexordering}.
\begin{algorithm}[h!]
\caption{Greedily generate an ordering based on the score function and the given starting vertex.}
\label{alg:generatevertexordering}
\begin{algorithmic}[1]
\Procedure{GenerateVertexOrdering}{$G, ScoreFunction, init$}
    \State{$Decomposition \gets (init)$}
    \State{$Left \gets \{init\}$}
    \State{$Right \gets V \setminus \{init\}$}
    \While{{$Right \neq \emptyset$}}
        \State{$Candidates \gets$ set returned by candidate set strategy}
        \If{there exists $v \in Candidates$ belonging to a trivial case}
            \State{$chosen \gets v$}
        \Else{}
            \State{$chosen \gets \argmin \limits_{v \in Candidates}(ScoreFunction(G, Left, Right, v))$}
        \EndIf{}

        \State{$Decomposition \gets Decomposition \cdot \{chosen\}$}
        \State{$Left \gets Left \cup \{chosen\}$}
        \State{$Right \gets Right \setminus \{chosen\}$}
    \EndWhile{}
\EndProcedure{}
\State{\textbf{return} $Decomposition$}
\end{algorithmic}
\end{algorithm}

\fi

At any point in the algorithm we denote the set of all vertices contained in the ordering by $Left$, and the remaining vertices by $Right$.
While $Right$ is not empty, we choose a vertex from a candidate set $Candidates \subseteq Right$,
 based on a set of trivial cases, and, if no trivial case applies, by making a local greedy choice using a score function that indicates the quality
of the current state $Left, Right$.

\ifarchive{}
\subsubsection{Selecting the initial vertex}
Selecting a good initial vertex can be of great influence on the quality of the decomposition.
Sharmin proposes to use a double breadth first search (BFS) in order to select the initial vertex. 
This is done by initiating a BFS, starting at an arbitrary vertex, after which a vertex of the last level of the BFS is selected.
This process is then repeated by using the found vertex as a starting point for the second BFS.
However, the fact that an arbitrary vertex is used for the first BFS already influences the boolean-width of the computed decomposition.
During our experiments we noticed that performing a single BFS sometimes gave better results.
But since we will see in Chapter~\ref{section:sigmarhoproblems} that applications are a lot more expensive in terms of running time, it is wise
to use all possible starting vertices when trying to find a good decomposition.
\fi

\subsubsection{Pruning}
Starting from multiple initial vertices allows us to do some pruning.
If we notice during the algorithm that the score of the decomposition that is being constructed exceeds the score of the best decomposition found so far,
we can stop immediately and move to the next initial vertex.
For this reason, it is wise to start with the most promising initial vertices (e.g.\ obtained by the double BFS method), and after that try all other initial vertices.

\subsubsection{Candidates}
The most straightforward choice for the set $Candidates$ is to take $Right$ entirely.
However, we may do unnecessary work here, since vertices that are more than 2 steps away from any vertex in $Left$ cannot decrease the size of $\un$.
This means that they should never be chosen by a greedy score function, which means that we can skip them right away.
By this reasoning, the set of $Candidates$ can be reduced to $N^2(Left) \cap Right = N(Left \cup N(Left)) \cap Right$.
Especially for larger sparse graphs, this can significantly decrease the running time.

\subsubsection{Trivial cases}
A vertex is chosen to be the next vertex in the ordering if it can be guaranteed that it is an optimal choice by means of a trivial case.
Lemma~\ref{lemma:trivialcases} generalizes results by Sharmin~\cite{Practicalaspects}, since the two trivial cases given by her are subcases of our lemma, 
namely $X = \emptyset$ and $X = \{u\}$ for all $u \in Left$.
Note that we can add a wide range of trivial cases by varying $X$, such as $X = Left$ and $\forall u, w \in Left: X = \{u, w\}$,
but this will increase the complexity of the algorithm. 

\begin{lemma}
[Appendix~\ref{prf:trivialcases}]
\label{lemma:trivialcases}
Let $X \subseteq Left$. If $\exists v \in Right$ such that $N(v) \cap Right = N(X) \cap Right$, 
then choosing $v$ will not change the boolean-width of the resulting decomposition.
\end{lemma}

\subsubsection{Relative Neighborhood Heuristic}

For a cut $(Left, Right)$ and a vertex $v$ define
\begin{align*}
Internal(v) &= (N(v) \cap N(Left)) \cap Right \\
External(v) &= (N(v) \setminus N(Left)) \cap Right
\end{align*}

In the original formulation by Sharmin~\cite{Practicalaspects} $\frac{|External(v)|}{|Internal(v)|}$ is used as a score function.
However, if we use $\frac{|External(v)|}{|Internal(v)| + |External(v)|} = \frac{|External(v)|}{|N(v) \cap Right|}$
we get the same ordering by Lemma~\ref{order-preserving}, without having an edge case for dividing by zero.
Furthermore, in contrast to Sharmin's proposal of checking for each vertex 
$w \in N(v)$ if $w \in N(Left) \cap Right$ or not, we can compute these sets directly by performing set operations.
We will refer to this heuristic by $\textsc{RelativeNeighborhood}$.

\begin{lemma}
[Appendix~\ref{prf:order-preserving}]
\label{order-preserving}
    The mapping $\frac{a}{b} \mapsto \frac{a}{a + b}$ is order preserving.
\end{lemma}

Two variations on this heuristic can be obtained through the score functions $\frac{|External(v)|}{|N(v)|}$ and $1 - \frac{|Internal(v)|}{|N(v)|}$, 
which work slightly better for sparse random graphs and extremely well for dense random graphs respectively.
We will refer to these two variations by $\textsc{RelativeNeighborhood}_2$ and $\textsc{RelativeNeighborhood}_3$.

One can easily see that the running time of these three algorithms is $O(n^3)$ and the required space amounts to $O(n)$.
Notice however that this algorithm only gives us a decomposition.
If we need to know the corresponding boolean-width we need to compute it afterwards, 
for instance by iteratively applying \Call{Increment-UN}{} on the vertices
in the decomposition, and taking the maximum value.
This would require an additional $O(n^2 \cdot 2^k)$ time and $O(n \cdot 2^k)$ space, where $k$ is the boolean-width of the decomposition.

\subsubsection{Least Cut Value Heuristic}

The \textsc{LeastCutValue} heuristic by Sharmin~\cite{Practicalaspects} greedily selects the next vertex $v \in Right$ 
that will have the smallest boolean dimension across the cut $(Left \cup \{v\}, Right \setminus \{v\})$.
This vertex is obtained by constructing the bipartite graph $BG = G[Left \cup \{v\}, Right \setminus \{v\}]$ for each $v \in Right$,
and counting the number of maximal independent sets of $BG$ using the $CCM_{IS}$~\cite{maxisshamin} algorithm on $BG$, 
with the time of $CCM_{IS}$ being exponential in $n$.

\subsubsection{Incremental Unions of Neighborhoods Heuristic}

Generating a bipartite graph and then calculating the number of maximal independent sets is a computational expensive approach.
A different way to compute the boolean dimension of each cut is by reusing the neighborhoods from the previous cut,
similarly to \textsc{Incremental-UN-exact}.
We present a new algorithm, called the \textsc{Incremental-UN-heuristic}, in Algorithm~\ref{alg:incremental-un-heuristic}.
A useful property of this algorithm is that the running time is output sensitive.
It follows that if a decomposition is not found within reasonable time, then the decomposition that would have been generated
is not useful for practical algorithms.

\begin{algorithm}
\caption{Greedy heuristic that incrementally keeps track of the Unions of Neighborhoods.}
\begin{algorithmic}[1]
\Procedure{Incremental-UN-Heuristic}{$G, init$}

\State{$Decomposition \gets (init)$}
\State{$Left, Right \gets \{init\}, V \setminus \{init\}$}
\State{$\un_{Left} \gets \{\emptyset, N(init) \cap Right\}$}

\While{{$Right \neq \emptyset$}}
    \State{$Candidates \gets$ set returned by candidate set strategy}
    \If{there exists $v \in Candidates$ belonging to a trivial case}
        \State{$chosen \gets v$}
        \State{$\un_{chosen} \gets $ \Call{Increment-UN}{$G, Left, \un_{Left}, v$}}
    \Else{}
        \ForAll{$v \in Candidates$}
            \State{$\un_v \gets $ \Call{Increment-UN}{$G, Left,\un_{Left}, v$}}
            \If{$chosen$ is undefined $\textbf{or}$ $|\un_v| < |\un_{chosen}|$}
                \State{$chosen \gets v$}
                \State{$\un_{chosen} \gets \un_v$}
            \EndIf{}
        \EndFor{}
    \EndIf{}

    \State{$Decomposition \gets Decomposition \cdot chosen$}
    \State{$Left \gets Left \cup \{chosen\}$}
    \State{$Right \gets Right \setminus \{chosen\}$}
    \State{$\un_{Left} \gets \un_{chosen}$}
\EndWhile{}
\State{\textbf{return} $Decomposition$}

\EndProcedure{}

\end{algorithmic}
\label{alg:incremental-un-heuristic}
\end{algorithm}

\begin{theorem}
[Appendix~\ref{prf:incrementalheuristic}]
\label{theorem:incrementalheuristic}
    The \textsc{Incremental-UN-heuristic} procedure runs in $O(n^3 \cdot 2^k)$ time using $O(n \cdot 2^k)$ space, 
	where $k$ is the boolean-width of the resulting linear decomposition.
\end{theorem}

\ifarchive{}

\subsubsection{Unsuccessful ideas}

\begin{itemize}
\item First Improvement --- Preliminary experiments pointed out that it not only gave worse results in terms of boolean-width, 
but it also increased the time needed to compute a decomposition, 
which can be explained by the output sensitivity of the \textsc{Incremental-UN-heuristic}.
In other words, even though the best improvement strategy takes more time to determine the next 
vertex for a single iteration, it is worthwhile to put effort in finding a good cut, as it also decreases the
time for future cuts.
\item Lookaheads --- This technique does not only look at the change of $\un$ resulting from choosing a candidate $v$, 
but also recursively considers the changes of the algorithm after $v$ has been chosen, up to a fixed depth.
With each level of depth added, the time complexity increases with a factor $n$, 
but experiments turned out that the benefits were only marginal.
\item Minimal Neighborhood Cover --- This heuristic tries to minimize the number of neighborhoods in $Left$ 
that are needed to cover the neighborhood of the vertex to be chosen.
\item Max Cardinality Search --- This heuristics selects vertices in such an order that at each step the vertex with most 
neighbors in $Left$ is chosen.
In practice this heuristic performed similar to other already known polynomial heuristics.
\end{itemize}

\fi

\section{Vertex subset problems}
\label{section:sigmarhoproblems}
Boolean decompositions can be used to efficiently solve a class of vertex subset problems 
called $(\sigma, \rho)$ vertex subset problems, which were introduced by Telle~\cite{Telle94complexityof}.
This class of problems consists of finding a $(\sigma, \rho)$-set of maximum or minimum cardinality
and contains well known problems such as the maximum independent set, 
the minimum dominating set and the maximum induced matching problem. 
The running time of the algorithm for solving these problems is $O(n^4 \cdot {nec_d(T, \delta)}^3)$~\cite{fastdynamicprogramming},
where $nec_d(T, \delta)$ is the number of equivalence classes of a problem specific equivalence relation, which can be bounded in terms of boolean-width.
In Section~\ref{section:experiments} we investigate how close the value of $nec_d(T, \delta)$ comes to any of the theoretical bounds. 

\subsection{Definitions}

\begin{definition}[$(\sigma, \rho)$-set]
Let $G = (V, E)$ be a graph. 
Let $\sigma$ and $\rho$ be finite or co-finite subsets of $\mathbb{N}$.
A subset $X \subseteq V$ is called a $(\sigma, \rho)$-set if the following holds
\[
	\forall v \in V : |N(v) \cap X| \in
	  \begin{cases}
			\sigma & \text{if } v \in X, \\
		\rho       & \text{if } v \in V \setminus X.
  \end{cases}
\]
\end{definition}

In order to confirm if a set $X$ is a $(\sigma, \rho)$-set we have to count the number of neighbors each vertex $v \in V$ has in $X$, where it suffices to count up until a certain number of neighbors.
As an example, when we want to confirm if a set $X$ is an independent set, which is equivalent to checking if $X$ is a $(\{0\}, \mathbb{N})$-set,
it is irrelevant if a vertex $v$ has more than one neighbor in $X$. 
We capture this property in the function $d: 2^\mathbb{N} \to \mathbb{N}$, which is defined as follows:

\begin{definition}[d-function]
Let $d(\mathbb{N}) = 0$. 
For every finite or co-finite set $\mu \subseteq \mathbb{N}$, let $d(\mu) = 1 + \min(\max\limits_{x\in \mathbb{N}} x : x \in \mu, \max\limits_{x \in \mathbb{N}} x : x \notin \mu)$.
Let $d(\sigma, \rho) = \max(d(\sigma), d(\rho))$.
\end{definition}

\begin{definition}[d-neighborhood]
Let $G = (V, E)$ be a graph. 
Let $A \subseteq V$ and $X \subseteq A$.
The \emph{d-neighborhood} of $X$ with respect to $A$, denoted by $N_A^d(X)$, is a multiset of vertices from $\overline{A}$, 
where a vertex $v \in \overline{A}$ occurs $\min(d, |N(v) \cap X|)$ times in $N_A^d(X)$.
A d-neighborhood can be represented as a vector of length $|\overline{A}|$ over $\{0, 1, \dots, d\}$.
\end{definition}

\begin{definition}[d-neighborhood equivalence]
Let $G = (V, E)$ be a graph and $A \subseteq V$. 
Two subsets $X, Y \subseteq A$ are said to be \emph{d-neighborhood equivalent} with respect to $A$, denoted by $X \equiv_A^d Y$, 
if it holds that $\forall v \in \overline{A}: \min(d, |X \cap N(v)|) = \min(d, |Y \cap N(v)|)$.
The number of equivalence classes of a cut $\cut{A}$ is denoted by $nec(\equiv_A^d)$.
The number of equivalence classes $nec_d(T, \delta)$ of a decomposition $(T, \delta)$ is defined as $\max(\nec{A},\nec{\overline{A}})$ over all cuts $\cut{A}$ of $(T, \delta)$.
\end{definition}

Note that $N_A^1(X) = N(X) \cap \overline{A}$. 
It can then be observed that $|\un(A)| = nec(\equiv_A^1)$~\cite[Theorem 3.5.5]{vatshelle}
Also note that $X \equiv_A^d Y$ if and only if $N_A^d(X) = N_A^d(Y)$. 

\subsection{Bounds on the number of equivalence classes}

We present a brief overview of the most relevant bounds that are currently known, 
for which we make use of a \emph{twin class partition} of a graph.

\begin{definition}[Twin class partition]
Let $G = (V, E)$ be a graph of size $n$ and let $A \subseteq V$.
The \emph{twin class partition} of $A$ is a partition of $A$ such that $\forall x, y \in A$, $x$ and $y$ are in the same partition class if and only if $N(x) \cap \overline{A} = N(y) \cap \overline{A}$.
The number of partition classes of $A$ is denoted by $ntc(A)$ and it holds that $ntc(A) \leq \min(n, 2^{\booldim(A)})$~\cite{boolean-width}.
\end{definition}

For all bounds listed below, let $G = (V, E)$ be a graph of size $n$ and let $d$ be a non-negative integer.
Let $\cut{A}$ be a cut induced by any node of a decomposition $(T, \delta)$ of $G$, and let $k = \booldim(A) = nec(\equiv_A^1)$. 

\begin{lemma}~\cite[Lemma 5]{fastdynamicprogramming}
\label{lemma:2dk}
$\nec{A} \leq 2^{d \cdot k^2}$.
\end{lemma}

\begin{lemma}~\cite[Lemma 5.2.2]{vatshelle}
$\nec{A} \leq {(d + 1)}^{\min(ntc(A), ntc(\overline{A}))}$.
\end{lemma}

\begin{lemma}[Appendix~\ref{prf:necntc}]
\label{lemma:necntc}
$\nec{A} \leq {ntc(A)}^{d \cdot k}$.
\end{lemma}

By Lemma~\ref{lemma:2dk} we conclude that we can solve $(\sigma, \rho)$ problems in $O^*(8^{dk^2})$. 
This shows that applications are more computationally expensive than using heuristics to find a decomposition.

\section{Experiments}
\label{section:experiments}

The experiments in this section are performed on a 64-bit Windows 7 computer, with a 3.40 GHz Intel Core i5-4670 CPU and 8GB of RAM. 
We implemented the algorithms using the C\# programming language and 
compiled our programs using the \emph{csc} compiler that comes with Visual Studio 12.0.\footnote{Source code of our implementations can be found on \url{https://github.com/Chiel92/boolean-width} and \url{https://github.com/FrankvH/BooleanWidth}}

\subsection{Comparing Heuristics on random graphs}
We will look at the performance of heuristics on randomly generated graphs, 
for which we used the Erd\"{o}s-R\'{e}nyi-model~\cite{erdos} to generate a fixed set of random graphs with varying edge probabilities.
\ifarchive{}
By using the same set of graphs for each heuristic, we rule out the possibility that one heuristic can get a slightly easier set of graphs than another.
\fi
In these experiments we start a heuristic once for each possible initial vertex, so $n$ times in total.
For the \textsc{RelativeNeighborhood} heuristic we select the best decomposition based upon the sum of the score function for all cuts, since
computing all actual linear boolean-width values would take $O(n^3 \cdot 2^k)$ time, thereby removing the purpose of this polynomial time heuristic.
For the set $Candidates$ we take $N^2(Left) \cap Right$, which avoids that we exclude certain optimal solutions, as opposed to Sharmin~\cite{Practicalaspects}, 
who restricted this set to $N(Left) \cap Right$. However, this does not affect the results significantly.

We let the edge probability vary between 0.05 and 0.95 with steps of size 0.05.
For each edge probability value, we generated 20 random graphs.
The result per edge probability is taken to be the average boolean-width over these 20 graphs, which are shown in Figure~\ref{fig:random_small}.
It can be observed that the \textsc{Incremental-UN-heuristic} procedure performs near optimal.
Furthermore we see that the \textsc{RelativeNeighborhood} variants perform somewhere in between the optimal value and the value of random decompositions.

\begin{figure}[!htb]
    \centering
    \includegraphics[width=\textwidth]{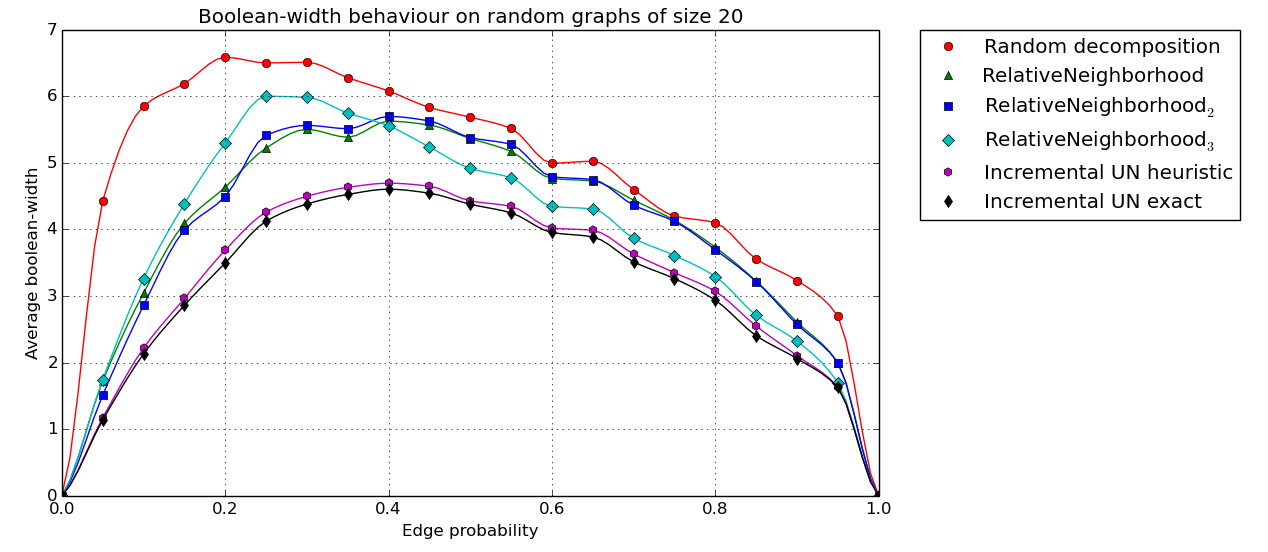}
    \caption{Performance of different heuristics on random generated graphs consisting of 20 vertices, with varying edge probabilities, in terms of linear boolean-width.}
\label{fig:random_small}
\end{figure}

\subsection{Comparing heuristics on real-world graphs}

In order to get an idea of how the \textsc{Incremental-UN-heuristic} compares to existing heuristics
we compare them by both the boolean-width of the generated decomposition and the time needed for computation.
We cannot compare the heuristics to the optimal solution, because computing an exact decomposition is not feasible on these graphs.
The graphs that were used come from $Treewidthlib$~\cite{twlib}, 
a collection of graphs that are used to benchmark algorithms using treewidth and related graph problems.

We ran the three different heuristics mentioned in Section~\ref{section:heuristics} with $Candidates = Right$ 
and with an additional two variations on the \textsc{Incremental-UN-heuristic} (IUN) by varying the set of start vertices.
The first variation, named 2-IUN, has two start vertices which are obtained through a single and double BFS respectively. 
The n-IUN heuristic uses all possible start vertices.
For all other heuristics we obtained the start vertex through performing a double BFS.
In Table~\ref{table:lbwheuristic} and~\ref{table:timeheuristic} we present the results of our experiments. 

\begin{table}[ht!]
\centering
\caption{\small{Linear boolean-width of the decompositions returned by different heuristics.}}
\begin{tabular}{|c|cc|ccccc|}
\hline
Graph & $|V|$ & Edge Density & Relative & LeastCut & IUN & 2-IUN & n-IUN  \\
\hline
barley & 48 & 0.11 & 5.70 & 5.91 & 5.91 & 4.70 & 4.58  \\
pigs-pp & 48 & 0.12 & 10.35 & 7.13 & 7.13 & 7.13 & 6.64  \\
david & 87 & 0.11 & 9.38 & 6.27 & 6.27 & 6.27 & 5.86  \\
celar04-pp & 114 & 0.08 & 11.67 & 7.27 & 7.27 & 7.27 & 7.27  \\
1bkb-pp & 127 & 0.18 & 16.81 & 9.98 & 9.98 & 9.53 & 9.53  \\
miles1500 & 128 & 0.64 & 8.17 & 5.58 & 5.58 & 5.58 & 5.29  \\
celar10-pp & 133 & 0.07 & 10.32 & 11.95 & 11.95 & 7.64 & 6.91  \\
munin2-pp & 167 & 0.03 & 15.17 & 9.61 & 9.61 & 9.61 & 7.61  \\
mulsol.i.5 & 186 & 0.23 & 7.55 & 5.29 & 5.29 & 5.29 & 3.58  \\
zeroin.i.2 & 211 & 0.16 & 7.92 & 4.46 & 4.46 & 4.46 & 3.81  \\
boblo & 221 & 0.01 & 19.00 & 4.32 & 4.32 & 4.32 & 4.00  \\
fpsol2.i-pp & 233 & 0.40 & 5.58 & 6.07 & 6.07 & 5.78 & 4.81  \\
munin4-wpp & 271 & 0.02 & 13.04 & 9.27 & 9.27 & 9.27 & 7.61  \\
\hline
\end{tabular}
\label{table:lbwheuristic}
\end{table}

\begin{table}[ht!]
\centering
\caption{\small{Time in seconds of the heuristics used to find linear boolean decompositions.}}
\begin{tabular}{|c|cc|ccccc|}
\hline
Graph & $|V|$ & Edge Density & Relative & LeastCut & IUN & 2-IUN & n-IUN  \\
\hline
barley & 48 & 0.11 & $<0.01$ & 0.18 & 0.01 & 0.02 & 0.16  \\
pigs-pp & 48 & 0.12 & $<0.01$ & 0.76 & 0.02 & 0.04 & 0.52  \\
david & 87 & 0.11 & 0.02 & 3.15 & 0.04 & 0.06 & 1.62  \\
celar04-pp & 114 & 0.08 & 0.04 & 5.73 & 0.14 & 0.23 & 9.85  \\
1bkb-pp & 127 & 0.18 & 0.06 & 198.05 & 1.14 & 4.18 & 107.32  \\
miles1500 & 128 & 0.64 & 0.06 & 44.57 & 0.10 & 0.14 & 7.05  \\
celar10-pp & 133 & 0.07 & 0.06 & 8.93 & 1.96 & 4.72 & 18.43  \\
munin2-pp & 167 & 0.03 & 0.11 & 3.81 & 0.80 & 3.37 & 30.21  \\
mulsol.i.5 & 186 & 0.23 & 0.09 & 37.88 & 0.13 & 0.27 & 8.80  \\
zeroin.i.2 & 211 & 0.16 & 0.06 & 18.70 & 0.09 & 0.11 & 5.85  \\
boblo & 221 & 0.01 & 0.29 & 3.39 & 0.28 & 0.56 & 46.22  \\
fpsol2.i-pp & 233 & 0.40 & 0.18 & 189.11 & 0.36 & 0.74 & 56.63  \\
munin4-wpp & 271 & 0.02 & 0.61 & 57.87 & 1.98 & 6.66 & 367.37  \\
\hline
\end{tabular}
\label{table:timeheuristic}
\end{table}

It is expected that the IUN heuristic and \textsc{LeastCutValue} heuristic give the same linear boolean-width, 
since both these heuristics greedily select the vertex that minimizes the boolean dimension.
The \textsc{RelativeNeighborhood} heuristic performs worse than all other heuristics in nearly all cases. 
While the difference might not seem very large, note that algorithms parameterized by boolean-width are exponential in the width of a decomposition.
The 2-IUN heuristic outperforms IUN in three cases while n-IUN gives a better decomposition in 11 out of 13 cases, 
which shows that a good initial vertex is of great influence on the width of the decomposition. 

Looking at the times displayed in Table~\ref{table:timeheuristic} for computing each decomposition 
we see that the \textsc{RelativeNeighborhood} heuristic is significantly faster.
This is to be expected because of the $O(n^3)$ time, compared to the exponential time for all other heuristics.
The interesting comparison that we can make is the difference between the IUN heuristic and \textsc{LeastCutValue} heuristic.
While both of these heuristics give the same decomposition, IUN is significantly faster. 
Additionally, even 2-IUN and n-IUN are often faster than the \textsc{LeastCutValue} heuristic.

\ifarchive{}
In Table~\ref{table:bigtable} (Appendix~\ref{appendix:tables}) we  show
linear boolean-width upperbounds that are obtained through using the IUN heuristic
with all starting vertices and $candidates = Right$.
We compare this with the best known tree-width and boolean-width values.
Examining these results, it seems that linear boolean-width seem to be more useful in practice
than boolean-width heuristics.
However, one should note that on certain graph classes, for instance graphs which look like
trees, boolean-width is a lot lower than linear boolean-width.
\fi

\subsection{Vertex subset experiments}

We have used the linear decompositions given by the n-IUN heuristic to compute the size of the 
maximum induced matching (MIM) in a selection of graphs, 
of which the results are presented in Table~\ref{table:sigmarhoexperiments}.
The maximum induced matching problem is defined as finding the largest $(\{1\}, \mathbb{N})$ set, with $d(\{1\}, \mathbb{N}) = 2$.
The choice for the MIM problem is arbitrary, any vertex subset problem with $d = 2$ 
will have the same number of equivalence classes and therefore they all require the same time when computing a solution.
We present the computed value of $nec_d(T, \delta)$, together with theoretical upperbounds.
For $d= 2$ a tight upperbound in terms of boolean-width is not known.
Note that we take the logarithm of each value, since we find this value easier to interpret and compare to other graph parameters.
We let $UB_1 = 2^{d \cdot \bw^2}$, $UB_2 = {(d+1)}^{\min ntc}$ and $UB_3 = ntc^{d \cdot \bw}$, 
with $ntc = \max\limits_{w \in T} ntc(\delta(w))$ and $\min ntc = \max\limits_{w \in T} \min(ntc(\delta(w)),ntc(\overline{\delta(w)}))$.

The column $MIM$ displays the size of the MIM in the graph, while the time column indicates the time needed to compute this set.
Missing values for $nec$ and MIM are caused by a lack of internal memory, because of the $O^*({nec_d(T, \delta)}^2)$ space requirement.
One can immediately see that there is a large gap between the upperbound for $nec_2$ in terms
of boolean-width and $nec_2$ itself.
Another interesting observation we can make by looking at the graphs zeroin.i.2 and boblo, 
is that a lower boolean-width does not imply a lower $nec_2$.
We even encountered this for decompositions of the same graph:
for the graph barley we observed $\bw(T, \delta) = 4.58$ and $\bw(T', \delta') = 4.81$,
while $\log_2(nec_2(T, \delta)) = 7.00$ and $\log_2(nec_2(T', \delta')) = 6.75$.
This suggests that this upperbound does not justify minimizing $nec_2$ through
boolean-width in practice.

\begin{table}[h!]
\centering
\caption{\small{Results of using the algorithm by Bui-Xuan et al.~\cite{fastdynamicprogramming} for solving $(\sigma, \rho)$ problems on graphs, 
using decompositions obtained through the n-IUN heuristic.}}
\begin{tabular}{|c|cc|ccc|cc|}
\hline
Graph & $\bw$ & $\log_2(nec)$ & $\log_2(UB_1)$ & $\log_2(UB_2)$  & $\log_2(UB_3)$  & $MIM$ & Time (s) \\
\hline
barley & 4.58 & 7.00 & 42.04 & 12.68 & 27.51 & 22 & 3  \\
pigs-pp & 6.64 & 10.31 & 88.28 & 19.02 & 49.17 & 22 & 1147  \\
david & 5.86 & 9.37 & 68.63 & 22.19 & 44.61 & 34 & 919  \\
celar04-pp & 7.27 & 11.15 & 105.61 & 28.53 & 65.74 & - & -  \\
1bkb-pp & 9.53 & - & 181.47 & 52.30 & 98.49 & - & -  \\
miles1500 & 5.29 & 9.30 & 55.87 & 34.87 & 49.69 & 8 & 4038  \\
celar10-pp & 6.91 & 10.34 & 95.41 & 25.36 & 59.70 & 50 & 10179  \\
munin2-pp & 7.61 & 11.82 & 115.97 & 19.02 & 54.60 & - & -  \\
mulsol.i.5 & 3.58 & 6.11 & 25.70 & 14.26 & 24.80 & 46 & 22  \\
zeroin.i.2 & 3.81 & 6.58 & 28.99 & 20.60 & 28.18 & 30 & 59  \\
boblo & 4.00 & 6.17 & 32.00 & 9.51 & 20.68 & 148 & 41  \\
fpsol2.i-pp & 4.81 & 8.07 & 46.22 & 22.19 & 36.61 & 46 & 934  \\
munin4-wpp & 7.61 & 12.13 & 115.97 & 19.02 & 57.98 & - & -  \\
\hline
\end{tabular}
\label{table:sigmarhoexperiments}
\end{table}

\section{Conclusion}
\label{section:conclusion}

We have presented a new heuristic and a new exact algorithm for finding linear boolean decompositions.
The heuristic has a running time that is several orders of magnitude lower than the previous best heuristic
and finds a decomposition in output sensitive time.
This means that if a decomposition is not found within reasonable time, then the decomposition that would have been generated
is not useful for practical algorithms.
Running the new heuristic once for every possible starting vertex results in significantly better decompositions compared to existing heuristics.

We have seen that if $\lbw(T, \delta) < \lbw(T', \delta')$, then there is no guarantee that $nec(T, \delta) < nec(T', \delta')$.
While in general it holds that minimizing boolean-width results in a low value of number of equivalence classes,
we think that it can be worthwhile to focus on minimizing the $nec_d$ instead of the boolean-width 
when solving vertex subset problems.
However, the number of equivalence classes is not symmetric, i.e., for a cut $\cut{A}$ $nec_d(A) \neq nec_d(\overline{A})$, 
which makes it harder to develop fast heuristics that focus on minimizing $nec_d$ since we need to keep track of both the equivalence classes of $A$ and $\overline{A}$.

Further research can be done in order to obtain even better heuristics and better upperbounds on both the linear 
boolean-width and boolean-width on graphs. 
For instance, combining properties of the \textsc{Incremental-UN-heuristic} and the \textsc{RelativeNeighborhood} heuristic 
might lead to better decompositions, as they make use of complementary features of a graph.
Another approach for obtaining good decompositions could be a branch and bound algorithm that makes us of trivial cases that are
used in the heuristics.
To decrease the time needed by the heuristics one can investigate reduction rules for linear boolean-width. 
While most reduction rules introduced by Sharmin~\cite{Practicalaspects} for boolean-width do not hold for linear boolean-width,
they can still be used on a graph after which we can use our heuristic on the reduced graph.
Although the resulting decomposition after reinserting the reduced vertices will not be linear, 
the asymptotic running time for applications does not increase~\cite{vhouten}. 
Another topic of research is to compare the performance of vertex subset algorithms parameterized by boolean-width 
to algorithms parameterized by treewidth~\cite{improvementconvolution}.

\bibliography{paper}{}

\begin{thebibliography}{10}

\bibitem{Belmonte201354}
R.~Belmonte and M.~Vatshelle.
\newblock Graph classes with structured neighborhoods and algorithmic
  applications.
\newblock {\em Theoretical Computer Science}, 511(0):54 -- 65, 2013.
\newblock Exact and Parameterized Computation.

\bibitem{boolean-width}
B.-M. Bui-Xuan, J.~A. Telle, and M.~Vatshelle.
\newblock Boolean-width of graphs.
\newblock In {\em IWPEC 2009}, volume 5917 of {\em LNCS}, pages 61--74.
  Springer, 2009.

\bibitem{fastdynamicprogramming}
B.-M. Bui-Xuan, J.~A. Telle, and M.~Vatshelle.
\newblock Fast dynamic programming for locally checkable vertex subset and
  vertex partitioning problems.
\newblock {\em Theoretical Computer Science}, 511:66--76, 2013.

\bibitem{dias}
V.~M.F. Dias, C.~M.H. de~Figueiredo, and J.~L. Szwarcfiter.
\newblock On the generation of bicliques of a graph.
\newblock {\em Discrete Applied Mathematics}, 155(14):1826 -- 1832, 2007.

\bibitem{erdos}
P.~Erd\"{o}s and A.~R\'{e}nyi.
\newblock On random graphs.
\newblock {\em Publicationes Mathematicae 6: 290–297}, 1959.

\bibitem{gaspers}
S.~Gaspers, D.~Kratsch, and M.~Liedloff.
\newblock On independent sets and bicliques in graphs.
\newblock In {\em Proc. WG 2008}, volume 5344 of {\em LNCS}, pages 171--182.
  Springer, 2008.

\bibitem{Hvidevold}
E.~M. Hvidevold, S.~Sharmin, J.~A. Telle, and M.~Vatshelle.
\newblock Finding good decompositions for dynamic programming on dense graphs.
\newblock In {\em IWPEC 2012}, volume 7112 of {\em LNCS}, pages 219--231.
  Springer, 2012.

\bibitem{booleanMatrixTheory}
K.~H. Kim.
\newblock {\em Boolean matrix theory and its applications (Monographs and
  textbooks in pure and applied mathematics)}.
\newblock Marcel Dekker, 1982.

\bibitem{maxisshamin}
F.~Manne and S.~Sharmin.
\newblock Efficient counting of maximal independent sets in sparse graphs.
\newblock In {\em Experimental Algorithms}, volume 7933 of {\em LNCS}, pages
  103--114. Springer, 2013.

\bibitem{UpperBoundsBoolw}
Y.~Rabinovich, J.~A. Telle, and M.~Vatshelle.
\newblock Upper bounds on boolean-width with applications to exact algorithms.
\newblock In {\em IWPEC 2013}, volume 8246 of {\em LNCS}, pages 308--320.
  Springer, 2013.

\bibitem{Pathwidth}
N.~Robertson and P.~D. Seymour.
\newblock Graph minors. {I}. {E}xcluding a forest.
\newblock {\em Journal of Combinatorial Theory, Series B}, 35(1):39 -- 61,
  1983.

\bibitem{Practicalaspects}
S.~Sharmin.
\newblock {\em Practical Aspects of the Graph Parameter Boolean-width}.
\newblock PhD thesis, University of Bergen, Norway, 2014.

\bibitem{Telle94complexityof}
J.~A. Telle.
\newblock Complexity of domination-type problems in graphs.
\newblock {\em Nordic Journal of Computing}, 1(1):157--171, 1994.

\bibitem{twlib}
Treewidthlib.
\newblock http://www.staff.science.uu.nl/{$\sim$}bodla101/treewidthlib/.
\newblock A benchmark for algorithms for treewidth and related graph problems.

\bibitem{vhouten}
F.~J.~P. van Houten.
\newblock Experimental research and algorithmic improvements involving the
  graph parameter boolean-width.
\newblock Master's thesis, Utrecht University, The Netherlands, 2015.

\bibitem{improvementconvolution}
J.~M.~M. van Rooij, H.~L. Bodlaender, and P.~Rossmanith.
\newblock Dynamic programming on tree decompositions using generalised fast
  subset convolution.
\newblock In {\em Algorithms - ESA 2009}, volume 5757 of {\em LNCS}, pages
  566--577. Springer, 2009.

\bibitem{vatshelle}
M.~Vatshelle.
\newblock {\em New width parameters of graphs}.
\newblock PhD thesis, University of Bergen, Norway, 2012.

\end{thebibliography}
\bibliographystyle{plain}

\newpage
\appendix

\section{Omitted proofs}

\subsection{Proof of Theorem~\ref{th:pathwidth}}
\label{prf:pathwidth}

\begin{claim}
For any graph $G$ it holds that $\lbw(G) \leq \pw(G) + 1$.
\end{claim}

\begin{proof}
We give a method of construction that 
gives us a linear boolean decomposition of a graph $G$ from a path decomposition of $G$.
Recall that a linear boolean decomposition can be defined through 
a linear ordering $\pi = \pi_1,\dots,\pi_n$ of $V$.
The idea is that given a path decomposition $X_1, \dots, X_n$ we select vertices one by one from a 
subset $X_i$ and append them to the linear ordering $\pi$, after which we move on to $X_{i + 1}$.
For shorthand notation we denote $\chi_i = \bigcup\limits_{j = 1}^i X_i$.

Let $S_i = \set{u}{u \in \chi_i: N(u) \cap \overline{\chi_i} \neq \emptyset}$.
For each $u \in S_i$ it holds that $\exists j > i\ \exists w \in X_{j}$ for which $\{u, w\} \in E$.
By definition of a path decomposition we know that there is a subset $X_{j}$ with $u, w \in X_j$, 
and since all subsets containing a certain vertex are subsequent in the path decomposition,
it follows that $u \in X_i$ and $u \in X_{i + 1}$, implying that $S_i \subseteq X_i$ and $S_i \subseteq X_{i + 1}$.
By definition, the unions of neighborhoods of $\chi_i$ can only consist of neighborhoods of subsets of $S_i$, 
thus it follows that $|\un(\chi_i)| = 2^{\booldim(\chi_i)} \leq 2^{|S_i|} \leq 2^{|X_i|} \leq 2^{\pw(G) + 1}$.
What remains to be shown is that while appending vertices one by one from a subset $X_{i + 1}$, 
the number of unions of neighborhoods will not exceed $2^{|X_{i + 1}|}$ at any point.
For each vertex $v \in X_{i + 1}$ there are two possibilities.
If $v \in S_i$, then appending $v$ to the linear ordering will not increase the boolean dimension, 
since $v$'s neighborhood was already an element of the unions of neighborhoods constructed so far.
If $v \notin S_i$, then it is possible that $v$ will contribute a new neighborhood to the unions of neighborhoods,
which will cause factor 2 increase in the worst case. 
There are at most $|X_{i + 1} \setminus S_i|$ such vertices, and because 
$S_i \subseteq X_{i + 1}$, it follows that $|X_{i + 1} \setminus S_i| = |X_{i + 1}| -  |S_i|$.
We conclude that at any point during construction it holds that
\[
\un(\chi_{i + 1}) = 2^{\booldim(\chi_{i + 1})} \leq 2^{|S_i|} \cdot 2^{|X_{i + 1}| -  |S_i|} = 2^{|X_{i + 1}|} \leq 2^{\pw(G) + 1}
\]
\end{proof}

\subsection{Adaptation of existing exact algorithms}
\label{subs:existing_exact_algorithms}
Straighforward dynamic programming leads to the following result.
\begin{theorem}
\label{thm:lbw_dp}
    A linear boolean decomposition of minimum boolean-width can be computed in $O(2.7284^n)$ time
    using $O(n \cdot 2^n)$ space.
\end{theorem}
\begin{proof}
    As a preprocessing step we compute for all cuts $A \subseteq V$ the values $|\un(A)|$ by computing $\#\mathcal{MIS}(G[A, \overline{A}])$.
    Computing $\#\mathcal{MIS}$ for any graph can be done in $O(1.3642^n)$ time~\cite{gaspers}.
    Doing this for all $A$ takes $O(2.7284^n)$ time.

    We solve recurrence relation~\eqref{eq:lboolw} in a bottom-up fashion.
    For each iteration, the minimum of $|A|$ numbers has to be taken.
    Suppose $|A| = k$, then this takes $O(k)$ time for each iteration.
    When solving the recurrence relation, $|A|$ goes from $1$ to $n$.
    Since there are $n \choose k$ subsets of size $k$, it takes $ \sum_{k=1}^{n} {n \choose k} k = O(n \cdot 2^{n - 1}) = O(n \cdot 2^n)$
    time to compute all values for $\lbw$.

    Because the preprocessing step of computing $\booldim$ is the bottleneck, the total time is $O(2.7284^n)$.
    The space requirements amount to $O(n \cdot 2^n)$, since $\booldim$ and $\lbw$ contain at most $2^n$ entries of integers of at most  $n$ bits.
\end{proof}

The currently fastest known exact algorithm for boolean-width runs in $O^*(2^{n+K})$~\cite{vatshelle},
where $K$ is a known upperbound for the boolean-width of the current graph.
By performing a binary search on $K$, we can achieve an output sensitive asymptotic running time.
Theorem~\ref{thm:lbw_vatshelle} is a direct adaptation to linear boolean-width.
\begin{theorem}
\label{thm:lbw_vatshelle}
    A linear boolean decomposition of minimum boolean-width for a graph $G$ can be computed in $O(n^3 \cdot 2^{n + \lbw(G)})$ time using $O(n \cdot 2^n)$ space.
\end{theorem}
\begin{proof}
    As a preprocessing step we compute for all cuts $A \subseteq V$ the values $|\un(A)|$, using a polynomial time delay algorithm,
    which lists maximal independent sets in $G[A, \overline{A}]$ with at most $O(n^3)$ time in between two results~\cite{dias}.
    We can use the upperbound $K$ as a limit for this algorithm, such that computing $\max(|\un(A)|, K)$ takes at most $O(n^3 \cdot K)$ time.

    Now consider relation~\eqref{eq:lboolw}. This can be solved in $O(n \cdot 2^n)$ time by the same reasoning as in Theorem~\ref{thm:lbw_dp}.
    This results in a total running time of $O(n^3 \cdot 2^{n + \lbw(G)})$ by binary search on $K$.
    The space requirements amount to $O(n \cdot 2^n)$, since the tables $\booldim$ and $\lbw$ contain at most $2^n$ entries of integers of at most $n$ bits.
\end{proof}

\subsection{Proof of Lemma~\ref{increment-un-correctness}}
\label{prf:increment-un-correctness}

\begin{claim}
The procedure Increment-UN is correct and runs in $O(n \cdot |\un_X|)$ time using $O(n \cdot |\un_X|)$ space.
\end{claim}

\begin{proof}
For proof by induction, assume that all unions of neighborhoods for the cut $(X, \overline{X})$ saved inside the set $\un_X$ are computed correctly.
For each neighborhood in $\un_X$ we only perform two actions to obtain new neighborhoods.
The first action is removing $v$, since $v$ cannot be in any neighborhood of $X \cup \set{v}$.
The second operation is adding $N(v)$ to an existing neighborhood, which also results in a valid new neighborhood across the cut.
It is clear that if a neighborhood is added to $\mathcal{U}$, then it is a valid neighborhood across the cut $(X \cup \{v\}, \overline{X} \setminus \{v\})$.
We now show that all valid neighborhoods of the cut $(X \cup \{v\}, \overline{X} \setminus \{v\})$ are contained in $\mathcal{U}$.
Assume for contradiction that $S$ is a valid neighborhood not contained in $\mathcal{U}$.
By definition, there is a set $R$ for which $N(R) \cap (\overline{X} \setminus \{v\}) = S$.
If $v \notin R$, then $N(R) \cap \overline{X} \in \un_X$, meaning that we add $N(R) \cap (\overline{X} \setminus \{v\})$ to $\mathcal{U}$, contradicting our assumption.
If $v \in R$, then $N(R \setminus \{v\}) \cap \overline{X} \in \un_X$.
During the algorithm we construct $(N(R \setminus \{v\}) \cup N(v)) \cap (\overline{X} \setminus \{v\})$, which is equal to $N(R) \cap (\overline{X} \setminus \{v\})$.
This means that $N(R) \cap (\overline{X} \setminus \{v\})$ is added to $\mathcal{U}$, also contradicting our assumption.
It follows that a neighborhood is contained in the set $\mathcal{U}$ if and only if it is a valid neighborhood across the cut $(X \cup \{v\}, \overline{X} \setminus \{v\})$.

The time is determined by the number of sets $S$ saved in $\un_X$.
The number of unions of neighborhoods that we iterate over does not exceed $|\un_X|$.
The set operations that are performed for each $S$ take at most $O(n)$ time.
This results in the total time for this algorithm to be $O(n \cdot |\un_X|)$.
The space requirements amount to $O(n \cdot |\un_X|)$, for storing $\mathcal{U}$ which contains at most $O(|\un_X|)$ sets of size at most $O(n)$.
\end{proof}

\subsection{Proof of Theorem~\ref{incremental-un-exact-correctness}}
\label{prf:incremental-un-exact-correctness}

\begin{claim}
    Given a graph $G$, Algorithm~\ref{alg:incremental-un-exact} can be used to compute $\lbw(G)$
    in $O(n \cdot 2^{n+\lbw(G)})$ time using $O(n \cdot 2^n)$ space.
\end{claim}

\begin{proof}
    Iteratively double $K$ in Algorithm~\ref{alg:incremental-un-exact}, starting with $K=1$, until it returns a number that is not $\infty$.
    By Lemma~\ref{lemma:incremental-un-exact-correctness} this will take $O(\sum_{\log K=1}^{\lbw(G)} n \cdot 2^{n+\log K}) = O(n \cdot 2^{n + \lbw(G) + 1}) = O(n \cdot 2^{n+\lbw(G)})$
    and take $O(n \cdot 2^n)$ space.
\end{proof}

\begin{lemma}
	\label{lemma:incremental-un-exact-correctness}
    Given a graph $G = (V,E)$ of size $n$ and an integer $K$, Algorithm~\ref{alg:incremental-un-exact}
    computes the linear boolean width, if it is at most $\log K$, in $O(n \cdot K \cdot 2^n)$ time
    using $O(n \cdot 2^n)$ space.
\end{lemma}
\begin{proof}
    Consider the first part of procedure \textsc{Incremental-UN-exact}, where the call to the procedure \textsc{Compute-count-UN} is made.
    It may not be immediately clear that $T_{\un}$ is always computed when necessary, since there
    may be $X$ such that $T_{\un}(X)$ is not computed, while $T_{\un}(X) \leq K$.
    Suppose that $X \subseteq V$ of size $i$ occurs in an optimal decomposition and $T_{\un}(X)$ has not been computed.
    Since we are dealing with linear decompositions, there exists an ordering $v_1,\dots,v_i$
    of $X$ such that for all $1 \leq j \leq i$, the set $X_j = \bigcup_{0 \leq j' \leq j} v_{j'}$ also occurs in the optimal decomposition.
    Obviously this implies that $T_{\un}(X_j) \leq K$ for all $j$.
    But this means that for all these $X_j$ the if-statement on line~\ref{booldim-pruning2}
    evaluates to true. But that means that $T_{\un}(X)$ must be computed, contradiction.
    Thus we conclude that $T_{\un}$ is computed correctly throughout the algorithm.
    The second part of procedure \textsc{Incremental-UN-exact} simply solves the recurrence in a bottom-up dynamic programming fashion.
    Finally, the procedure \textsc{Increment-UN} is correct by Lemma~\ref{increment-un-correctness}.

    We now analyze the running time.
    Consider the procedure \textsc{Compute-count-UN}.
    We observe that the procedure can only be called once for each $X \subseteq V$,
    because as soon as the call is made, $T_{\un}(X)$ will be defined
    and line~\ref{booldim-pruning1} prevents further calls with equal $X$.
    At every call the for-loop has to make at most $n$ iterations,
    thus we obtain $O(n \cdot 2^n)$ iterations in total.
    If line~\ref{booldim-pruning1} evaluates false, the body of the for-loop takes constant time.
    If line~\ref{booldim-pruning1} evaluates true, the call to \textsc{Increment-UN} takes $O(n \cdot 2^K)$ time (by Lemma~\ref{increment-un-correctness}),
    as $|\un_X| \leq K$ (otherwise by line~\ref{booldim-pruning2} the call to \textsc{Compute-count-UN} would not have been made).
    Because line~\ref{booldim-pruning1} only returns true at most $O(2^n)$ times,
    the time of \textsc{Compute-count-UN} amounts to $O(n \cdot 2^{n+K})$.
    Consider the rest of the code in \textsc{Incremental-UN-exact}.
    The three outer for-loops account for $n\cdot 2^n$ executions of the inner code block,
    which take $O(1)$ time, resulting in $O(n \cdot 2^n)$ time in total.
    Thus, in total the time amounts $O(n \cdot 2^{n+K})$.

    For the space requirements, we observe that the tables $T_{\un}$ and $S$ are of size at most $2^n$ storing numbers of $n$ bits.
    Moreover, the recursion of \textsc{Compute-count-UN} can be at most $n$ deep, so only $n$ unions
    of neighborhoods have to be stored, which are at most of size $n \cdot 2^K$.
    Since $O(n \cdot 2^K) \subseteq O(n \cdot 2^{n/2}) \subsetneq O(n \cdot 2^n)$, the total space requirements amount to $O(n \cdot 2^n)$.
\end{proof}

\subsection{Proof of Lemma~\ref{lemma:trivialcases}}
\label{prf:trivialcases}

\begin{claim}
Let $X \subseteq Left$. If $\exists v \in Right$ such that $N(v) \cap Right = N(X) \cap Right$, 
then choosing $v$ will not change the boolean-width of the resulting decomposition.
\end{claim}

\begin{proof}
The choice for $v$ will not change the unions of neighborhoods in any way, which means that $\un(Left) = \un(Left \cup \{v\})$.
Thus, for any vertex in $Right \setminus \{v\}$ it will hold that it will interact in the exact same with with $\un(Left)$ as it would with $\un(Left \cup \{v\})$, resulting in the boolean dimension of the computed ordering being the same.
\end{proof}

\subsection{Proof of Lemma~\ref{order-preserving}}
\label{prf:order-preserving}

\begin{claim}
    The mapping $\frac{a}{b} \mapsto \frac{a}{a + b}$ is order preserving.
\end{claim}

\begin{proof}
    Suppose $\frac{a}{b} \leq \frac{c}{d}$. Then $ad - bc \leq 0$.
    Now we see that
    \[
    \frac{a}{a+b} - \frac{c}{c+d}
    = \frac{a(c+d) - c(a+b)}{(c+d)(a+b)}
    = \frac{ac + ad - ac - bc}{(c+d)(a+b)}
    = \frac{ad - bc}{(c+d)(a+b)}
    \leq 0
    \]
    Thus $ \frac{a}{a+b} \leq \frac{c}{c+d}$.
\end{proof}

\subsection{Proof of Theorem~\ref{theorem:incrementalheuristic}}
\label{prf:incrementalheuristic}

\begin{claim}
    The \textsc{Incremental-UN-heuristic} procedure runs in $O(n^3 \cdot 2^k)$ time using $O(n \cdot 2^k)$ space, 
	where $k$ is the boolean-width of the resulting linear decomposition.
\end{claim}

\begin{proof}
The time is determined by the number of sets saved in $\un_{Left}$.
The worst case consisting of $Candidates = Right$ will result in at most $n$ iterations and calls to \Call{Increment-UN}{}.
This call takes $O(n \cdot 2^{|\un_{Left}|})$ time by Lemma~\ref{increment-un-correctness}.
By definition $|\un_{Left}|$ never exceeds $2^k$, where $k$ is the boolean-width of the resulting decomposition. Because we need to make $n$ greedy choices to process the
entire graph, we conclude that the total time for this algorithm is $O(n^3 \cdot 2^k)$
For the space requirements we observe that all structures in the algorithm require $O(n)$ space, except for the unions of neighborhoods.
Since there are only stored two of them at any time and they require at most $O(n \cdot 2^k)$ space, the total space requirements amount to $O(n \cdot 2^k)$.
\end{proof}

\subsection{Proof of Lemma~\ref{lemma:necntc}}
\label{prf:necntc}
\begin{claim}
$\nec{A} \leq {ntc(A)}^{d \cdot k}$.
\end{claim}
\begin{proof}
We make use of a graph parameter called \emph{maximum induced matching-width}~\cite{Belmonte201354}.
Let $mim(A)$ denote the maximum matching-width of $A$.
It has been shown that for a graph $G$ and for any subset $A \subseteq V$ it holds that $mim(A) \leq \booldim(A)$~\cite[Theorem 4.2.10]{vatshelle}.
From~\cite[Lemma 5.2.3]{vatshelle} we know that $nec(\equiv_A^d) \leq {ntc(A)}^{d \cdot mim(A)}$, thus $nec(\equiv_A^d) \leq {ntc(A)}^{d \cdot k}$.
\end{proof}

\newpage
\section{Figures and Tables}

\subsection{Figures}
\begin{figure}[!htb]
    \centering
    \includegraphics[width=\textwidth]{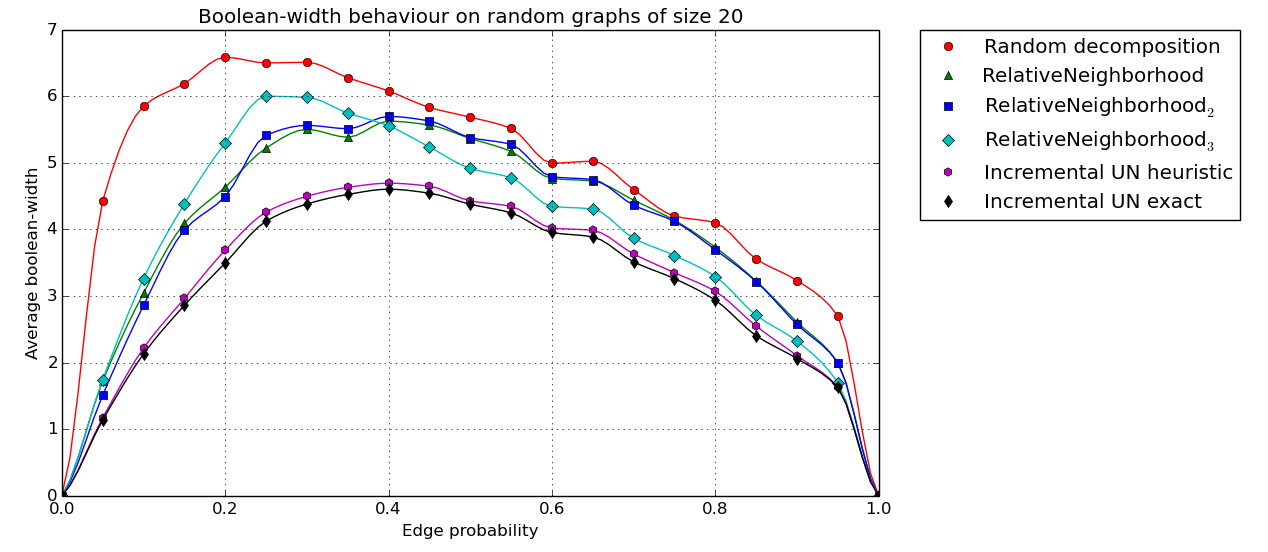}
    \caption{Performance of different heuristics on random generated graphs consisting of 20 vertices, with varying edge probabilities, in terms of linear boolean-width.}
\label{fig:random_smallappendixappendix}
\end{figure}

\begin{figure}[!htb]
    \centering
    \includegraphics[width=\textwidth]{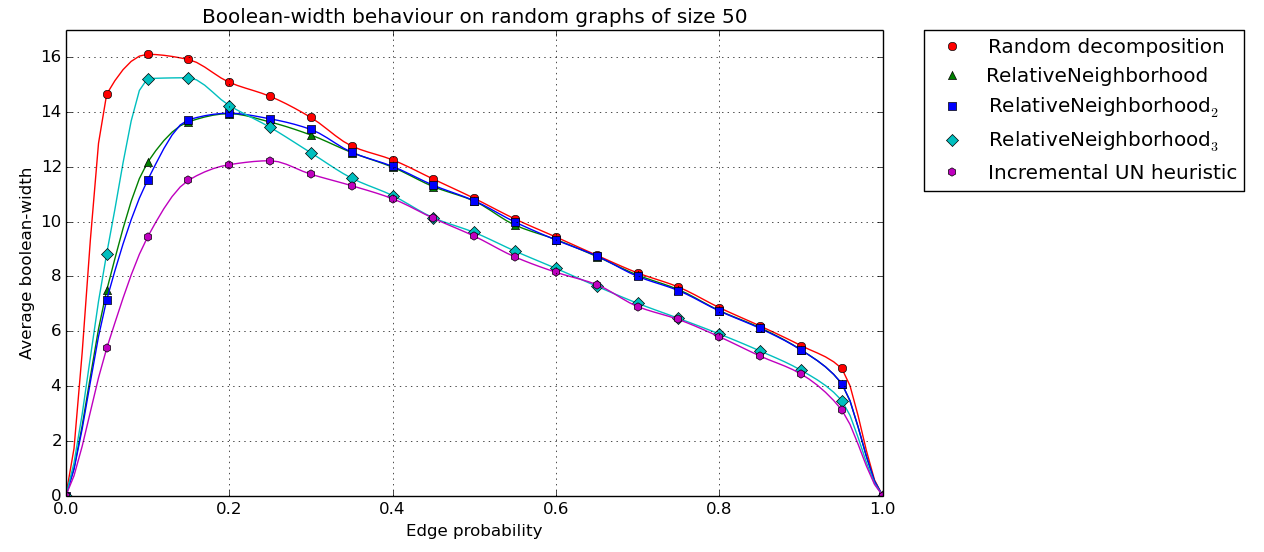}
    \caption{
	Performance of different heuristics on random generated graphs consisting of 50 vertices, with varying edge probabilities.
	Because of feasibility limitations, the \textsc{Incremental-UN-exact} algorithm is only used for the in Figure~\ref{fig:random_large}.
	While the optimal values are now unknown, it is clear that \textsc{Incremental-UN-heuristic} outperforms all other heuristics.
	Interestingly enough, \textsc{RelativeNeighborhood\textsubscript{3}} peers with \textsc{Incremental-UN-heuristic} as soon as the edge probability exceeds 0.4.
	Moreover, \textsc{RelativeNeighborhood} and \textsc{RelativeNeighborhood\textsubscript{2}} do not perform better than a random decomposition generator after
	the edge probability exceeds 0.4.
	We also observe that the highest boolean-width values are reached when the edge probability is around 0.1--0.2,
	indicating that the size of the graphs has an influence on the edge-probability-boolean-width-curve.
	Also note that it seems that dense random graphs have lower linear boolean-width than sparse graphs.
	Therefore it may be profitable to use \textsc{RelativeNeighborhood\textsubscript{3}} when dense graphs are encountered.}
\label{fig:random_large}
\end{figure}

\subsection{Tables}
\label{appendix:tables}

\begin{center}
\begin{longtable}{|c|cc|ccccc|}
\caption{Linear boolean-width of the decompositions returned by the heuristics described in Section~\ref{section:heuristics}, with $Candidates = Right$.
For 2-IUN we use two start vertices; one is obtained through a single BFS search, while the other is obtained through a double BFS search.
The n-IUN heuristic uses all $n$ start vertices, and all other heuristics use start vertices obtained through performing a double BFS.}
\label{table:fullmediumheuristiclbw}
\\
\hline
Graph & $|V|$ & Edge Density & Relative & LeastCut & IUN & 2-IUN & n-IUN  \\
\hline
\endfirsthead
\multicolumn{8}{c}%
{\tablename\ \thetable\ -- \textit{Continued from previous page}} \\
\hline
Graph & $|V|$ & Edge Density & Relative & LeastCut & IUN & 2-IUN & n-IUN  \\
\hline
\endhead
\hline \multicolumn{8}{r}{\textit{Continued on next page}} \\
\endfoot
\hline
\endlastfoot
alarm & 37 & 0.10 & 3.32 & 3.00 & 3.00 & 3.00 & 3.00  \\
barley & 48 & 0.11 & 5.70 & 5.91 & 5.91 & 4.70 & 4.58  \\
pigs-pp & 48 & 0.12 & 10.35 & 7.13 & 7.13 & 7.13 & 6.64  \\
BN\_100 & 58 & 0.17 & 15.84 & 11.56 & 11.56 & 10.86 & 10.86  \\
eil76 & 76 & 0.08 & 8.86 & 8.33 & 8.33 & 8.33 & 8.33  \\
david & 87 & 0.11 & 9.38 & 6.27 & 6.27 & 6.27 & 5.86  \\
1jhg & 101 & 0.17 & 12.86 & 8.67 & 8.67 & 8.49 & 8.41  \\
1aac & 104 & 0.25 & 20.29 & 12.40 & 12.40 & 12.40 & 12.33  \\
celar04-pp & 114 & 0.08 & 11.67 & 7.27 & 7.27 & 7.27 & 7.27  \\
1a62 & 122 & 0.21 & 18.92 & 11.68 & 11.68 & 11.28 & 11.14  \\
1bkb-pp & 127 & 0.18 & 16.81 & 9.98 & 9.98 & 9.53 & 9.53  \\
1dd3 & 128 & 0.17 & 16.61 & 9.98 & 9.98 & 9.90 & 9.90  \\
miles1500 & 128 & 0.64 & 8.17 & 5.58 & 5.58 & 5.58 & 5.29  \\
miles250 & 128 & 0.05 & 7.95 & 7.13 & 7.13 & 5.39 & 4.58  \\
celar10-pp & 133 & 0.07 & 10.32 & 11.95 & 11.95 & 7.64 & 6.91  \\
anna & 138 & 0.05 & 12.65 & 8.67 & 8.67 & 8.51 & 7.94  \\
pr152 & 152 & 0.04 & 12.69 & 11.19 & 11.19 & 10.36 & 8.29  \\
munin2-pp & 167 & 0.03 & 15.17 & 9.61 & 9.61 & 9.61 & 7.61  \\
mulsol.i.5 & 186 & 0.23 & 7.55 & 5.29 & 5.29 & 5.29 & 3.58  \\
zeroin.i.2 & 211 & 0.16 & 7.92 & 4.46 & 4.46 & 4.46 & 3.81  \\
boblo & 221 & 0.01 & 19.00 & 4.32 & 4.32 & 4.32 & 4.00  \\
fpsol2.i-pp & 233 & 0.40 & 5.58 & 6.07 & 6.07 & 5.78 & 4.81  \\
munin4-wpp & 271 & 0.02 & 13.04 & 9.27 & 9.27 & 9.27 & 7.61  \\
\end{longtable}
\end{center}

\begin{center}
\begin{longtable}{|c|cc|ccccc|}
\caption{Time in seconds of the heuristics used to find the linear boolean decompositions of which the boolean-width is displayed in Table~\ref{table:fullmediumheuristiclbw}.}\\
\hline
Graph & $|V|$ & Edge Density & Relative & LeastCut & IUN & 2-IUN & n-IUN  \\
\hline
\endfirsthead
\multicolumn{8}{c}%
{\tablename\ \thetable\ -- \textit{Continued from previous page}} \\
\hline
Graph & $|V|$ & Edge Density & Relative & LeastCut & IUN & 2-IUN & n-IUN  \\
\hline
\endhead
\hline \multicolumn{8}{r}{\textit{Continued on next page}} \\
\endfoot
\hline
\endlastfoot
alarm & 37 & 0.10 & $<0.01$ & 0.02 & $<0.01$ & $<0.01$ & 0.06  \\
barley & 48 & 0.11 & $<0.01$ & 0.18 & 0.01 & 0.02 & 0.16  \\
pigs-pp & 48 & 0.12 & $<0.01$ & 0.76 & 0.02 & 0.04 & 0.52  \\
BN\_100 & 58 & 0.17 & $<0.01$ & 25.10 & 0.41 & 1.24 & 17.17  \\
eil76 & 76 & 0.08 & 0.02 & 5.00 & 0.13 & 0.29 & 8.35  \\
david & 87 & 0.11 & 0.02 & 3.15 & 0.04 & 0.06 & 1.62  \\
1jhg & 101 & 0.17 & 0.03 & 24.46 & 0.21 & 0.48 & 14.75  \\
1aac & 104 & 0.25 & 0.04 & 754.54 & 5.66 & 11.81 & 375.31  \\
celar04-pp & 114 & 0.08 & 0.04 & 5.73 & 0.14 & 0.23 & 9.85  \\
1a62 & 122 & 0.21 & 0.06 & 585.95 & 3.10 & 11.57 & 376.26  \\
1bkb-pp & 127 & 0.18 & 0.06 & 198.05 & 1.14 & 4.18 & 107.32  \\
1dd3 & 128 & 0.17 & 0.07 & 117.21 & 0.92 & 2.74 & 91.19  \\
miles1500 & 128 & 0.64 & 0.06 & 44.57 & 0.10 & 0.14 & 7.05  \\
miles250 & 128 & 0.05 & 0.02 & 0.56 & 0.05 & 0.10 & 1.24  \\
celar10-pp & 133 & 0.07 & 0.06 & 8.93 & 1.96 & 4.72 & 18.43  \\
anna & 138 & 0.05 & 0.06 & 20.81 & 0.22 & 0.57 & 19.95  \\
pr152 & 152 & 0.04 & 0.10 & 50.74 & 1.76 & 5.66 & 120.06  \\
munin2-pp & 167 & 0.03 & 0.11 & 3.81 & 0.80 & 3.37 & 30.21  \\
mulsol.i.5 & 186 & 0.23 & 0.09 & 37.88 & 0.13 & 0.27 & 8.80  \\
zeroin.i.2 & 211 & 0.16 & 0.06 & 18.70 & 0.09 & 0.11 & 5.85  \\
boblo & 221 & 0.01 & 0.29 & 3.39 & 0.28 & 0.56 & 46.22  \\
fpsol2.i-pp & 233 & 0.40 & 0.18 & 189.11 & 0.36 & 0.74 & 56.63  \\
munin4-wpp & 271 & 0.02 & 0.61 & 57.87 & 1.98 & 6.66 & 367.37  \\
\end{longtable}
\end{center}

\begin{center}
\begin{longtable}{|c|cc|ccc|cc|}
\caption{Results of using the algorithm by Bui-Xuan et al.~\cite{fastdynamicprogramming} 
for solving $(\sigma, \rho)$ problems on graphs, 
using decompositions obtained using the IUN heuristic using all starting vertices.
The columns $UB$ indicate theoretical upperbounds on the number of equivalence classes, with
$UB_1 = 2^{d \cdot \bw^2}$, $UB_2 = {(d+1)}^{\min ntc}$ and $UB_3 = ntc^{d \cdot \bw}$, 
with $ntc = \max\limits_{w \in T} ntc(\delta(w))$ and $\min ntc = \max\limits_{w \in T} \min(ntc(\delta(w)),ntc(\overline{\delta(w)}))$.}\\
\hline
Graph & $\bw$ & $\log_2(nec)$ & $\log_2(UB_1)$ & $\log_2(UB_2)$  & $\log_2(UB_3)$  & $MIM$ & Time (s) \\
\hline
\endfirsthead
\multicolumn{8}{c}%
{\tablename\ \thetable\ -- \textit{Continued from previous page}} \\
\hline
Graph & $\bw$ & $\log_2(nec)$ & $\log_2(UB_1)$ & $\log_2(UB_2)$  & $\log_2(UB_3)$  & $MIM$ & Time (s) \\
\hline
\endhead
\hline \multicolumn{8}{r}{\textit{Continued on next page}} \\
\endfoot
\hline
\endlastfoot
alarm & 3.00 & 4.32 & 18.00 & 7.92 & 13.93 & 18 & $<1$  \\
barley & 4.58 & 7.00 & 42.04 & 12.68 & 27.51 & 22 & 3  \\
pigs-pp & 6.64 & 10.31 & 88.28 & 19.02 & 49.17 & 22 & 1147  \\
BN\_100 & 10.86 & - & 235.93 & 36.45 & 105.53 & - & -  \\
eil76 & 8.33 & 12.63 & 138.81 & 22.19 & 65.10 & - & -  \\
david & 5.86 & 9.37 & 68.63 & 22.19 & 44.61 & 34 & 919  \\
1jhg & 8.41 & 13.53 & 141.58 & 41.21 & 81.75 & - & -  \\
1aac & 12.33 & - & 304.08 & 72.91 & 141.25 & - & -  \\
celar04-pp & 7.27 & 11.15 & 105.61 & 28.53 & 65.74 & - & -  \\
1a62 & 11.14 & - & 248.09 & 60.23 & 121.61 & - & -  \\
1bkb-pp & 9.53 & - & 181.47 & 52.30 & 98.49 & - & -  \\
1dd3 & 9.90 & - & 196.11 & 52.30 & 103.17 & - & -  \\
miles1500 & 5.29 & 9.30 & 55.87 & 34.87 & 49.69 & 8 & 4038  \\
miles250 & 4.58 & 7.24 & 42.04 & 15.85 & 31.72 & 52 & 37  \\
celar10-pp & 6.91 & 10.34 & 95.41 & 25.36 & 59.70 & 50 & 10179  \\
anna & 7.94 & 11.94 & 125.98 & 33.28 & 75.48 & - & -  \\
pr152 & 8.29 & 12.76 & 137.45 & 22.19 & 63.13 & - & -  \\
munin2-pp & 7.61 & 11.82 & 115.97 & 19.02 & 54.60 & - & -  \\
mulsol.i.5 & 3.58 & 6.11 & 25.70 & 14.26 & 24.80 & 46 & 22  \\
zeroin.i.2 & 3.81 & 6.58 & 28.99 & 20.60 & 28.18 & 30 & 59  \\
boblo & 4.00 & 6.17 & 32.00 & 9.51 & 20.68 & 148 & 41  \\
fpsol2.i-pp & 4.81 & 8.07 & 46.22 & 22.19 & 36.61 & 46 & 934  \\
munin4-wpp & 7.61 & 12.13 & 115.97 & 19.02 & 57.98 & - & -  \\
\end{longtable}
\end{center}

\begin{table}[H]
\centering
\caption{\small{Width of linear boolean decompositions found with the IUN heuristic using the start vertices returned by performing a double BFS, and with $candidates = N^2(Left) \cap Right$ in order to decrease the computation time.
The values of the two others heuristics are taken from~\cite{Practicalaspects}.
Missing entries are caused by a lack of internal memory which is caused by the $O(n \cdot 2^k)$ space requirement, with $k$ being the linear boolean-width of the computed decomposition. 
The last column indicates the time of the IUN heuristic.}}
\begin{tabular}{|c|cc|ccc|c|}
\hline
Graph & $|V|$ & Edge Density & LeastUncommon & Relative & IUN & Time (s)  \\
\hline
link-pp & 308 & 0.02 & 34.81 & 28.68 & 17.44 & 610.09  \\
diabetes-wpp & 332 & 0.01 & 8.58 & 18.58 & 5.32 & 1.53  \\
link-wpp & 339 & 0.02 & 35.00 & 29.03 & 16.79 & 374.04  \\
celar10 & 340 & 0.02 & 20.81 & 15.00 & 10.17 & 1.83  \\
celar11 & 340 & 0.02 & 19.54 & 14.70 & 10.80 & 1.88  \\
rd400 & 400 & 0.01 & 34.73 & 21.32 & 17.01 & 1,007.03  \\
diabetes & 413 & 0.01 & 29.32 & 19.32 & - & -  \\
fpsol2.i.3 & 425 & 0.10 & 15.87 & 8.92 & 7.67 & 2.11  \\
pigs & 441 & 0.01 & 24.04 & 18.00 & 12.39 & 20.08  \\
celar08 & 458 & 0.02 & 24.95 & 15.00 & 10.17 & 2.12  \\
d493 & 493 & 0.01 & 20.29 & 48.10 & 16.73 & 708.57  \\
homer & 561 & 0.01 & 36.22 & 28.49 & - & -  \\
rat575 & 575 & 0.01 & 16.48 & 37.23 & - & -  \\
u724 & 724 & 0.01 & 18.72 & 50.09 & - & -  \\
inithx.i.1 & 864 & 0.05 & 11.98 & 7.22 & 6.81 & 7.31  \\
munin2 & 1003 & $<0.01$ & 31.25 & 12.13 & 11.91 & 61.17  \\
vm1084 & 1084 & $<0.01$ & 15.21 & 48.95 & - & -  \\
BN\_24 & 1819 & $<0.01$ & 4.91 & 2.32 & 2.58 & 610.72  \\
BN\_25 & 1819 & $<0.01$ & 4.64 & 2.32 & 2.58 & 601.41  \\
BN\_23 & 2425 & $<0.01$ & 8.48 & 3.17 & 2.58 & 1,808.29  \\
BN\_26 & 3025 & $<0.01$ & 6.98 & 2.32 & 3.58 & 4,532.83  \\
\hline
\end{tabular}
\label{table:table}
\end{table}

\begin{center}
\begin{longtable}{|c|cc|ccc|c|}
\caption{Linear boolean-width upperbounds that are obtained through using the IUN heuristic with all starting vertices and $candidates = Right$.
The $tw$ column gives an upperbound on the treewidth, while the $bw$ column gives an upperbound on the boolean-width, which values are taken from~\cite{Practicalaspects}.
Cursive graph names marked with an asterisk indicate the graphs for which, in theory, the linear boolean decomposition 
will give a higher bound on the running time than the boolean decomposition, i.e., graphs for which $2^{2 lbw} > 2^{3 bw}$.
From this it seems that linear boolean-width seem to be more useful in practice than
boolean-width heuristics.
However, one should note that on certain graph classes, for instance graphs which look like
trees, boolean-width is a lot lower than linear boolean-width.
}\\
\hline
Graph & $|V|$ & Edge Density & $tw$ & $bw$ & $lbw$ & $lbw/bw$  \\
\hline
\endfirsthead
\multicolumn{7}{c}%
{\tablename\ \thetable\ -- \textit{Continued from previous page}} \\
\hline
Graph & $|V|$ & Edge Density & $tw$ & $bw$ & $lbw$ & $lbw/bw$  \\
\hline
\endhead
\hline \multicolumn{7}{r}{\textit{Continued on next page}} \\
\endfoot
\hline
\endlastfoot
celar06-pp-003 & 4 & 0.5 & 2 & 1 & 1 & 1.00   \\
\emph{diabetes-pp-001*} & 6 & 0.8 & 4 & 1 & 1.58 & 1.58   \\
\emph{munin3-pp-001*} & 7 & 0.81 & 5 & 1 & 1.58 & 1.58   \\
\emph{munin3-pp-002*} & 7 & 0.81 & 5 & 1 & 1.58 & 1.58   \\
celar06-pp-000 & 8 & 0.43 & 3 & 1 & 1 & 1.00   \\
diabetes-pp-002 & 8 & 0.61 & 4 & 2.32 & 2.32 & 1.00   \\
mainuk-pp & 9 & 0.78 & 6 & 1.58 & 1.58 & 1.00   \\
rl5934-pp-001 & 10 & 0.44 & 4 & 2.81 & 3.17 & 1.13   \\
fl3795-pp-001 & 10 & 0.44 & 4 & 2.81 & 3 & 1.07   \\
fl3795-pp-003 & 10 & 0.44 & 4 & 2.81 & 3 & 1.07   \\
fl3795-pp-002 & 10 & 0.44 & 4 & 2.81 & 3.17 & 1.13   \\
pathfinder-pp-001 & 11 & 0.58 & 5 & 2.58 & 3.32 & 1.29   \\
myciel3 & 11 & 0.36 & 5 & 3 & 3.46 & 1.15   \\
pcb3038-pp-001 & 11 & 0.4 & 5 & 3 & 2.81 & 0.94   \\
fl3795-pp-004 & 11 & 0.42 & 4 & 3 & 3.46 & 1.15   \\
pathfinder-pp & 12 & 0.65 & 6 & 2.58 & 2.81 & 1.09   \\
celar11-pp-002 & 13 & 0.59 & 7 & 2.81 & 3.17 & 1.13   \\
celar04-pp-001-000 & 15 & 0.74 & 9 & 1.58 & 2 & 1.27   \\
weeduk & 15 & 0.47 & 7 & 1.58 & 1.58 & 1.00   \\
fungiuk & 15 & 0.34 & 4 & 2 & 1.58 & 0.79   \\
pcb3038-pp-002 & 15 & 0.3 & 5 & 3 & 2.81 & 0.94   \\
mildew-wpp & 15 & 0.3 & 4 & 2.58 & 3.32 & 1.29   \\
celar04-pp-001 & 16 & 0.78 & 10 & 1.58 & 2 & 1.27   \\
celar06-pp & 16 & 0.84 & 11 & 1.58 & 1.58 & 1.00   \\
celar10-pp-001 & 16 & 0.51 & 8 & 3 & 3.46 & 1.15   \\
celar09-pp-001 & 16 & 0.51 & 8 & 3 & 3.17 & 1.06   \\
celar08-pp-002 & 16 & 0.51 & 8 & 3 & 3.32 & 1.11   \\
celar07-pp-002 & 16 & 0.45 & 7 & 3 & 3.32 & 1.11   \\
barley-pp-001 & 16 & 0.42 & 7 & 3.32 & 3.32 & 1.00   \\
celar11-pp-004 & 16 & 0.36 & 6 & 3.17 & 3.58 & 1.13   \\
munin2-pp-005 & 16 & 0.3 & 5 & 3 & 3.58 & 1.19   \\
munin2-pp-006 & 16 & 0.3 & 5 & 3 & 3.58 & 1.19   \\
munin2-pp-003 & 16 & 0.3 & 5 & 3.17 & 3.7 & 1.17   \\
munin2-pp-004 & 16 & 0.3 & 5 & 3.17 & 3.7 & 1.17   \\
munin2-pp-007 & 17 & 0.35 & 7 & 3.46 & 3.58 & 1.03   \\
munin2-pp-011 & 17 & 0.35 & 7 & 3.46 & 3.58 & 1.03   \\
munin2-pp-010 & 17 & 0.35 & 7 & 3.46 & 3.81 & 1.10   \\
munin2-pp-008 & 17 & 0.35 & 7 & 3.46 & 3.58 & 1.03   \\
munin2-pp-009 & 18 & 0.31 & 6 & 3.46 & 3.81 & 1.10   \\
munin2-pp-012 & 18 & 0.31 & 6 & 3.46 & 3.81 & 1.10   \\
celar01-pp-002 & 19 & 0.65 & 10 & 2 & 2.32 & 1.16   \\
celar02-pp & 19 & 0.67 & 10 & 2 & 2 & 1.00   \\
celar05-pp-001 & 19 & 0.66 & 11 & 2 & 2.32 & 1.16   \\
celar11-pp-001 & 19 & 0.65 & 10 & 2 & 2.32 & 1.16   \\
fl3795-pp-005 & 19 & 0.22 & 4 & 3.32 & 3.58 & 1.08   \\
water-pp-001 & 21 & 0.45 & 9 & 3.81 & 4.09 & 1.07   \\
anna-pp & 22 & 0.64 & 12 & 3.46 & 3.81 & 1.10   \\
water-pp & 22 & 0.42 & 9 & 4.17 & 4.32 & 1.04   \\
water-wpp & 22 & 0.42 & 9 & 4.17 & 4.32 & 1.04   \\
munin4-pp-001 & 23 & 0.26 & 8 & 3.58 & 4 & 1.12   \\
munin4-pp-002 & 23 & 0.26 & 8 & 3.58 & 4 & 1.12   \\
myciel4 & 23 & 0.28 & 10 & 5 & 5.49 & 1.10   \\
BN\_29 & 24 & 0.18 & 5 & 2 & 2.32 & 1.16   \\
BN\_28 & 24 & 0.18 & 5 & 2 & 2.32 & 1.16   \\
queen5\_5 & 25 & 0.53 & 18 & 5.29 & 5.67 & 1.07   \\
barley-pp & 26 & 0.24 & 7 & 3.7 & 3.46 & 0.94   \\
fl3795-pp-006 & 26 & 0.16 & 5 & 3.81 & 4.17 & 1.09   \\
david-pp & 29 & 0.47 & 13 & 4.09 & 4.32 & 1.06   \\
barley-wpp & 29 & 0.2 & 7 & 3.81 & 3.58 & 0.94   \\
pcb3038-pp-003 & 29 & 0.12 & 5 & 4.32 & 4.75 & 1.10   \\
celar02-wpp & 30 & 0.33 & 10 & 2.81 & 2.58 & 0.92   \\
water & 32 & 0.25 & 9 & 4.39 & 4.75 & 1.08   \\
BN\_16-pp-015 & 34 & 0.28 & 11 & 3.58 & 4.39 & 1.23   \\
celar06-wpp & 34 & 0.28 & 11 & 3 & 3.17 & 1.06   \\
BN\_16-pp-014 & 34 & 0.28 & 11 & 3.81 & 4.86 & 1.28   \\
1bx7-pp & 34 & 0.31 & 11 & 4.7 & 4.39 & 0.93   \\
mildew & 35 & 0.13 & 4 & 3 & 3.32 & 1.11   \\
queen6\_6 & 36 & 0.46 & 25 & 7.65 & 8.08 & 1.06   \\
alarm & 37 & 0.1 & 4 & 2.58 & 3 & 1.16   \\
celar03-pp-001 & 38 & 0.34 & 14 & 5.81 & 6.11 & 1.05   \\
\emph{munin4-pp-003*} & 38 & 0.16 & 8 & 3.58 & 5.39 & 1.51   \\
munin4-pp-004 & 38 & 0.16 & 8 & 4.17 & 5.39 & 1.29   \\
celar08-pp-001 & 39 & 0.38 & 16 & 5.09 & 5.21 & 1.02   \\
oesoca & 39 & 0.09 & 3 & 2.32 & 3 & 1.29   \\
1bx7 & 41 & 0.24 & 11 & 4.91 & 4.75 & 0.97   \\
oesoca42 & 42 & 0.08 & 3 & 2.32 & 3.17 & 1.37   \\
celar07-pp-001 & 45 & 0.32 & 16 & 5.46 & 5.86 & 1.07   \\
celar01-pp-001 & 47 & 0.25 & 15 & 5.88 & 6.36 & 1.08   \\
celar05-pp-002 & 47 & 0.25 & 15 & 6.07 & 5.83 & 0.96   \\
myciel5 & 47 & 0.22 & 19 & 8.12 & 6.49 & 0.80   \\
1ubq-pp & 47 & 0.16 & 12 & 5.95 & 8.79 & 1.48   \\
pigs-pp-001 & 47 & 0.12 & 9 & 5.95 & 7.07 & 1.19   \\
1brf-pp & 48 & 0.36 & 22 & 7.01 & 7.25 & 1.03   \\
1rb9 & 48 & 0.37 & 22 & 6.77 & 7.17 & 1.06   \\
celar11-pp-003 & 48 & 0.23 & 15 & 5.73 & 4.58 & 0.80   \\
\emph{mainuk*} & 48 & 0.18 & 7 & 3.58 & 6.49 & 1.81   \\
barley & 48 & 0.11 & 7 & 4 & 3.7 & 0.93   \\
pigs-pp & 48 & 0.12 & 9 & 5.7 & 6.64 & 1.16   \\
1brf & 49 & 0.35 & 22 & 7.01 & 7.3 & 1.04   \\
queen7\_7 & 49 & 0.4 & 35 & 10.36 & 10.97 & 1.06   \\
1kth-pp & 51 & 0.33 & 20 & 7.06 & 5.86 & 0.83   \\
1i07-pp & 51 & 0.28 & 15 & 5.55 & 7.18 & 1.29   \\
eil51.tsp & 51 & 0.11 & 9 & 5.78 & 5.78 & 1.00   \\
1igq-pp & 52 & 0.37 & 23 & 6.74 & 7.45 & 1.11   \\
1kth & 52 & 0.32 & 20 & 7.04 & 6.87 & 0.98   \\
1g6x & 52 & 0.31 & 19 & 6.89 & 7.21 & 1.05   \\
1igq & 54 & 0.35 & 23 & 6.89 & 7.61 & 1.10   \\
zeroin.i.1-pp & 54 & 0.89 & 46 & 1.58 & 1.58 & 1.00   \\
1e0b-pp & 55 & 0.33 & 24 & 7.69 & 8.32 & 1.08   \\
munin4-pp-006 & 55 & 0.11 & 8 & 4.32 & 5.17 & 1.20   \\
munin4-pp-005 & 55 & 0.11 & 8 & 4.39 & 5.17 & 1.18   \\
1j75 & 56 & 0.36 & 27 & 8.51 & 8.94 & 1.05   \\
1k61-pp & 56 & 0.37 & 26 & 8.02 & 8.37 & 1.04   \\
1sem-pp & 56 & 0.37 & 26 & 8.09 & 8.5 & 1.05   \\
1bbz-pp & 56 & 0.35 & 25 & 8.18 & 8.36 & 1.02   \\
1bf4-pp & 57 & 0.39 & 26 & 7.63 & 7.79 & 1.02   \\
1cka & 57 & 0.38 & 27 & 8.55 & 8.87 & 1.04   \\
1sem & 57 & 0.36 & 26 & 8.32 & 8.66 & 1.04   \\
zeroin.i.2-pp & 57 & 0.69 & 32 & 2.81 & 3.32 & 1.18   \\
zeroin.i.3-pp & 57 & 0.69 & 32 & 3 & 3.32 & 1.11   \\
1bbz & 57 & 0.34 & 25 & 8.3 & 8.36 & 1.01   \\
1oai-pp & 57 & 0.32 & 22 & 7.94 & 8.28 & 1.04   \\
1jo8 & 58 & 0.37 & 27 & 8.46 & 8.73 & 1.03   \\
1oai & 58 & 0.32 & 22 & 7.87 & 8.15 & 1.04   \\
celar01-pp-003 & 58 & 0.19 & 15 & 6.97 & 6.89 & 0.99   \\
1g2b-pp & 59 & 0.37 & 28 & 8.5 & 8.99 & 1.06   \\
1igd-pp & 59 & 0.36 & 25 & 7.66 & 7.9 & 1.03   \\
1kq1-pp & 59 & 0.35 & 27 & 8.63 & 8.94 & 1.04   \\
1pwt-pp & 59 & 0.38 & 29 & 8.85 & 9.24 & 1.04   \\
1i07 & 59 & 0.23 & 15 & 5.52 & 5.93 & 1.07   \\
1k61 & 60 & 0.33 & 26 & 8.32 & 8.81 & 1.06   \\
1kq1 & 60 & 0.34 & 27 & 8.79 & 8.89 & 1.01   \\
1ku3-pp & 60 & 0.33 & 23 & 7.46 & 7.53 & 1.01   \\
1e0b & 60 & 0.29 & 24 & 8.13 & 8.42 & 1.04   \\
knights8\_8-pp & 60 & 0.09 & 16 & 10.77 & 11.3 & 1.05   \\
1gut-pp & 61 & 0.33 & 22 & 7.19 & 7.54 & 1.05   \\
1i2t & 61 & 0.35 & 27 & 8.38 & 9.03 & 1.08   \\
1igd & 61 & 0.34 & 25 & 7.75 & 7.9 & 1.02   \\
1pwt & 61 & 0.36 & 29 & 8.81 & 9.27 & 1.05   \\
1ku3 & 61 & 0.32 & 23 & 7.53 & 7.61 & 1.01   \\
1g2b & 62 & 0.34 & 28 & 8.72 & 9.05 & 1.04   \\
1fr3-pp & 62 & 0.32 & 21 & 7.16 & 7.29 & 1.02   \\
celar04-pp-002 & 62 & 0.17 & 16 & 6.86 & 7.26 & 1.06   \\
1bf4 & 63 & 0.34 & 26 & 7.9 & 8.09 & 1.02   \\
1r69 & 63 & 0.35 & 30 & 9.12 & 9.51 & 1.04   \\
munin1-pp-001 & 63 & 0.09 & 11 & 5.58 & 6.43 & 1.15   \\
1gcq-pp & 64 & 0.36 & 30 & 8.95 & 9.38 & 1.05   \\
queen8\_8 & 64 & 0.36 & 45 & 13.16 & 14.05 & 1.07   \\
1a8o & 64 & 0.27 & 25 & 9.11 & 9.3 & 1.02   \\
knights8\_8 & 64 & 0.08 & 16 & 11.06 & 11.64 & 1.05   \\
1fjl & 65 & 0.29 & 26 & 7.9 & 8.49 & 1.07   \\
1c9o & 66 & 0.34 & 29 & 8.75 & 8.88 & 1.01   \\
1hg7 & 66 & 0.33 & 29 & 8.81 & 9.13 & 1.04   \\
1ezg & 66 & 0.25 & 23 & 8.33 & 7 & 0.84   \\
1en2-pp & 66 & 0.21 & 17 & 7.46 & 8.54 & 1.14   \\
munin1-pp & 66 & 0.09 & 11 & 5.58 & 6.43 & 1.15   \\
1c4q & 67 & 0.34 & 31 & 9.45 & 9.71 & 1.03   \\
1fse & 67 & 0.33 & 27 & 8.58 & 8.75 & 1.02   \\
1kw4 & 67 & 0.3 & 28 & 9.39 & 5.73 & 0.61   \\
1gut & 67 & 0.28 & 22 & 7.47 & 7.36 & 0.99   \\
1fr3 & 67 & 0.28 & 21 & 7.29 & 7.47 & 1.02   \\
1b67-pp & 67 & 0.25 & 16 & 6.61 & 9.61 & 1.45   \\
1gcq & 68 & 0.33 & 30 & 9.36 & 9.65 & 1.03   \\
1ail-pp & 68 & 0.28 & 24 & 8.11 & 8.33 & 1.03   \\
1d3b-pp & 68 & 0.3 & 25 & 8.54 & 5.78 & 0.68   \\
1b67 & 68 & 0.25 & 16 & 6.61 & 8.52 & 1.29   \\
1c75 & 69 & 0.29 & 30 & 9.88 & 8.31 & 0.84   \\
1ail & 69 & 0.27 & 24 & 8.07 & 9.68 & 1.20   \\
1d3b & 69 & 0.29 & 25 & 8.44 & 8.53 & 1.01   \\
1en2 & 69 & 0.2 & 17 & 7.24 & 7 & 0.97   \\
1cc8 & 70 & 0.34 & 32 & 9.35 & 9.63 & 1.03   \\
1dj7-pp & 70 & 0.3 & 27 & 8.12 & 8.22 & 1.01   \\
1i27-pp & 70 & 0.3 & 27 & 8.67 & 8.82 & 1.02   \\
1l9l & 70 & 0.29 & 29 & 9.26 & 10 & 1.08   \\
1ljo-pp & 71 & 0.31 & 30 & 8.92 & 9.02 & 1.01   \\
1dp7-pp & 71 & 0.3 & 27 & 9.21 & 9.15 & 0.99   \\
graph03-pp-001 & 71 & 0.11 & 20 & 12.53 & 12.24 & 0.98   \\
1mgq-pp & 72 & 0.31 & 28 & 8.98 & 9.08 & 1.01   \\
1i27 & 73 & 0.28 & 27 & 8.78 & 9.06 & 1.03   \\
mulsol.i.1-pp & 73 & 0.83 & 50 & 2.32 & 2.58 & 1.11   \\
1dj7 & 73 & 0.28 & 27 & 9.66 & 8.22 & 0.85   \\
1ldd & 74 & 0.31 & 32 & 9.6 & 9.73 & 1.01   \\
1ljo & 74 & 0.29 & 30 & 8.88 & 9.06 & 1.02   \\
1mgq & 74 & 0.3 & 28 & 8.91 & 9.06 & 1.02   \\
huck & 74 & 0.11 & 10 & 2.81 & 3.32 & 1.18   \\
1ubq & 74 & 0.08 & 12 & 6.61 & 7.75 & 1.17   \\
1ig5 & 75 & 0.29 & 33 & 10.45 & 10.64 & 1.02   \\
1dp7 & 76 & 0.27 & 27 & 9.01 & 9.3 & 1.03   \\
celar10-pp-002 & 76 & 0.15 & 16 & 7.25 & 6.58 & 0.91   \\
celar08-pp-003 & 76 & 0.15 & 16 & 7.41 & 6.58 & 0.89   \\
celar09-pp-002 & 76 & 0.15 & 16 & 7.46 & 6.58 & 0.88   \\
1iqz & 77 & 0.29 & 33 & 10 & 10.1 & 1.01   \\
1qtn-pp & 77 & 0.25 & 24 & 8.56 & 8.33 & 0.97   \\
\emph{munin3-pp-003*} & 79 & 0.09 & 7 & 4.17 & 12.73 & 3.05   \\
graph03-pp & 79 & 0.1 & 20 & 12.99 & 5.61 & 0.43   \\
sodoku-elim1 & 80 & 0.28 & 45 & 9.47 & 12 & 1.27   \\
\emph{jean*} & 80 & 0.08 & 9 & 3.91 & 6.54 & $1.67$   \\
celar05-pp & 80 & 0.13 & 15 & 7.2 & 4.58 & 0.64   \\
sodoku & 81 & 0.25 & 45 & 9 & 12.7 & 1.41   \\
celar03-pp & 81 & 0.13 & 14 & 6.19 & 6.11 & 0.99   \\
graph03-wpp & 84 & 0.09 & 20 & 12.74 & 12.92 & 1.01   \\
1fk5 & 85 & 0.23 & 31 & 10.76 & 10.1 & 0.94   \\
1aba & 85 & 0.25 & 29 & 10.13 & 10.81 & 1.07   \\
graph01-pp-001 & 85 & 0.09 & 24 & 13.4 & 13.66 & 1.02   \\
1ctj-pp & 86 & 0.25 & 33 & 10.78 & 11.07 & 1.03   \\
1ctj & 87 & 0.25 & 33 & 10.74 & 11.04 & 1.03   \\
1ptf & 87 & 0.3 & 38 & 11.21 & 10.86 & 0.97   \\
1qtn & 87 & 0.21 & 24 & 9.15 & 8.97 & 0.98   \\
david & 87 & 0.11 & 13 & 5.32 & 5.86 & 1.10   \\
graph05-pp-001 & 87 & 0.1 & 24 & 12.68 & 13.31 & 1.05   \\
1awd & 89 & 0.28 & 38 & 10.8 & 11.13 & 1.03   \\
celar03-wpp & 89 & 0.11 & 14 & 6.17 & 6.49 & 1.05   \\
celar05-wpp & 89 & 0.11 & 15 & 7.52 & 6.54 & 0.87   \\
graph01-pp & 89 & 0.08 & 24 & 14.62 & 13.96 & 0.95   \\
munin1-wpp & 90 & 0.05 & 11 & 7.23 & 7.58 & 1.05   \\
1jhg-pp & 91 & 0.19 & 25 & 8.34 & 8.41 & 1.01   \\
graph05-pp & 91 & 0.1 & 24 & 13.84 & 13.49 & 0.97   \\
celar07-pp & 92 & 0.12 & 16 & 6 & 6 & 1.00   \\
a280.tsp-pp & 92 & 0.06 & 14 & 8.23 & 7.38 & 0.90   \\
\emph{kroE100.tsp-pp*} & 92 & 0.06 & 10 & 6.48 & 14.84 & 2.29   \\
1g2r-pp & 93 & 0.26 & 37 & 11.87 & 11.51 & 0.97   \\
graph01-wpp & 93 & 0.07 & 24 & 14.69 & 11.41 & 0.78   \\
1czp & 94 & 0.27 & 38 & 11.47 & 11.6 & 1.01   \\
1g2r & 94 & 0.25 & 37 & 12.17 & 14.19 & 1.17   \\
graph05-wpp & 94 & 0.09 & 24 & 14.38 & 13.18 & 0.92   \\
1c5e & 95 & 0.26 & 36 & 11.06 & 10.83 & 0.98   \\
myciel6 & 95 & 0.17 & 35 & 13.4 & 7.86 & 0.59   \\
homer-pp & 95 & 0.17 & 31 & 14.61 & 13.88 & 0.95   \\
kroA100.tsp-pp & 95 & 0.06 & 10 & 7.61 & 6.58 & 0.86   \\
celar11-pp & 96 & 0.1 & 15 & 6.64 & 5.98 & 0.90   \\
munin3-pp & 96 & 0.07 & 7 & 4.32 & 5.86 & 1.36   \\
celar07-wpp & 97 & 0.01 & 16 & 6 & 7.17 & 1.20   \\
\emph{kroC100.tsp-pp*} & 97 & 0.06 & 10 & 6.94 & 11.97 & 1.72   \\
1plc & 98 & 0.25 & 35 & 11.28 & 11.1 & 0.98   \\
1lkk-pp & 99 & 0.24 & 34 & 11 & 10.84 & 0.99   \\
1d4t-pp & 99 & 0.23 & 35 & 11.88 & 6.58 & 0.55   \\
celar11-wpp & 99 & 0.1 & 15 & 7.17 & 4.91 & 0.68   \\
1i0v & 100 & 0.24 & 41 & 12.21 & 12.47 & 1.02   \\
celar02 & 100 & 0.06 & 10 & 3.32 & 4.91 & 1.48   \\
\emph{celar06*} & 100 & 0.07 & 11 & 3.81 & 14.85 & 3.90   \\
graph05 & 100 & 0.08 & 24 & 13.7 & 13.36 & 0.98   \\
graph01 & 100 & 0.07 & 24 & 14.61 & 14.21 & 0.97   \\
graph03 & 100 & 0.07 & 20 & 13.29 & 8.41 & 0.63   \\
1erv & 101 & 0.25 & 41 & 12.26 & 12.44 & 1.01   \\
1jhg & 101 & 0.17 & 25 & 8.87 & 11.97 & 1.35   \\
1iib-pp & 102 & 0.27 & 40 & 11.98 & 11.76 & 0.98   \\
1d4t & 102 & 0.22 & 35 & 12.87 & 10.31 & 0.80   \\
1iib & 103 & 0.26 & 40 & 12.62 & 11.79 & 0.93   \\
1b0n & 103 & 0.19 & 32 & 10.81 & 11.17 & 1.03   \\
1lkk & 103 & 0.22 & 34 & 11.89 & 13.56 & 1.14   \\
1aac & 104 & 0.25 & 41 & 12.29 & 12.33 & 1.00   \\
1bkf-pp & 105 & 0.23 & 36 & 11.1 & 11.4 & 1.03   \\
1bkf & 106 & 0.23 & 36 & 11.69 & 11.44 & 0.98   \\
1bkr & 107 & 0.24 & 44 & 14.4 & 13.75 & 0.95   \\
1rro & 107 & 0.23 & 43 & 15.36 & 3.58 & 0.23   \\
1f9m & 109 & 0.23 & 45 & 14.27 & 13.56 & 0.95   \\
\emph{pathfinder*} & 109 & 0.04 & 6 & 3.32 & 10.83 & 3.26   \\
celar04-pp & 110 & 0.09 & 16 & 7.29 & 7.27 & 1.00   \\
1fs1 & 114 & 0.21 & 34 & 13.79 & 7.36 & 0.53   \\
celar04-wpp & 116 & 0.07 & 16 & 7.95 & 11.1 & 1.40   \\
1gef-pp & 117 & 0.22 & 43 & 12.93 & 13.35 & 1.03   \\
1gef & 119 & 0.21 & 43 & 13.6 & 13.35 & 0.98   \\
mulsol.i.5-pp & 119 & 0.36 & 31 & 3 & 3 & 1.00   \\
1a62-pp & 120 & 0.21 & 37 & 14.7 & 11.14 & 0.76   \\
1a62 & 122 & 0.21 & 37 & 13.62 & 9.68 & 0.71   \\
1dd3-pp & 124 & 0.17 & 31 & 14.6 & 9.25 & 0.63   \\
ch130.tsp-pp & 125 & 0.05 & 12 & 8.67 & 9.53 & 1.10   \\
1bkb-pp & 127 & 0.18 & 30 & 15.55 & 9.9 & 0.64   \\
miles1500 & 128 & 0.64 & 77 & 4.86 & 5.29 & 1.09   \\
1dd3 & 128 & 0.17 & 31 & 11.68 & 4.58 & 0.39   \\
miles500 & 128 & 0.14 & 22 & 9.42 & 7.04 & 0.75   \\
\emph{miles250*} & 128 & 0.05 & 9 & 4.95 & 9.61 & 1.94   \\
1bkb & 131 & 0.17 & 30 & 14.53 & 6.91 & 0.48   \\
celar10-pp & 133 & 0.07 & 16 & 9.08 & 7.7 & 0.85   \\
anna & 138 & 0.04 & 12 & 6.67 & 7.25 & 1.09   \\
celar09-wpp & 142 & 0.06 & 16 & 8.49 & 7 & 0.82   \\
celar01-pp & 157 & 0.07 & 15 & 7.39 & 7 & 0.95   \\
celar01-wpp & 158 & 0.06 & 15 & 7.09 & 7.61 & 1.07   \\
munin2-pp & 167 & 0.03 & 7 & 5.49 & 6.91 & 1.26   \\
mulsol.i.3 & 184 & 0.23 & 32 & 4.95 & 3.58 & 0.72   \\
mulsol.i.4 & 185 & 0.23 & 32 & 4.81 & 3.58 & 0.74   \\
mulsol.i.5 & 186 & 0.23 & 31 & 4.95 & 3.58 & 0.72   \\
mulsol.i.2 & 188 & 0.22 & 32 & 4.81 & 3.58 & 0.74   \\
celar08-wpp & 190 & 0.05 & 16 & 9.64 & 11.48 & 1.19   \\
mulsol.i.1 & 197 & 0.2 & 50 & 4 & 4.17 & 1.04   \\
zeroin.i.3 & 206 & 0.17 & 32 & 5.39 & 3.81 & 0.71   \\
zeroin.i.1 & 211 & 0.19 & 50 & 3.7 & 3.32 & 0.90   \\
zeroin.i.2 & 211 & 0.16 & 32 & 5.39 & 3.81 & 0.71   \\
fpsol2.i.1-pp & 233 & 0.4 & 66 & 4.91 & 4.81 & 0.98   \\
\label{table:bigtable}
\end{longtable}
\end{center}

\end{document}